\documentclass[
reprint,
aps,pra,
superscriptaddress,
floatfix,
]{revtex4-2}

\usepackage{graphicx}
\usepackage{dcolumn}
\usepackage{amsmath,amssymb,amsfonts,amsthm,mathtools,bm}
\usepackage{braket}
\usepackage{enumitem}
\usepackage{xcolor}
\usepackage[colorlinks=true,linkcolor=blue,citecolor=blue,urlcolor=blue]{hyperref}
\usepackage[nameinlink,capitalise]{cleveref}
\newtheorem{theorem}{Theorem}

\newtheorem{lemma}{Lemma}
\newtheorem{remark}{Remark}
\newtheorem{corollary}{Corollary}

\newcommand{\R}{\mathbb R}
\newcommand{\C}{\mathbb C}
\newcommand{\Cop}{\mathcal C}
\newcommand{\E}{\mathbb E}
\newcommand{\ii}{\mathrm i}
\newcommand{\dd}{\,\mathrm d}
\newcommand{\Tr}{\operatorname{Tr}}

\newcommand{\Lop}{\mathcal L}

\newcommand{\Dop}{\mathcal D}

\newcommand{\ip}[2]{\left\langle #1,#2\right\rangle}

\begin{document}

\title{From nonlinear stochastic differential equations to quantum channels: the Kolmogorov--Lindblad mapping}

\author{Hsuan-Cheng Wu}
    \email{wu.hsuancheng@psu.edu}
\affiliation{Department of Mathematics, The Pennsylvania State University, University Park, PA 16802, USA}
\author{Xiantao Li}
    \email{xiantao.li@psu.edu}
\affiliation{Department of Mathematics, The Pennsylvania State University, University Park, PA 16802, USA}

\date{\today}

\begin{abstract}
Nonlinear stochastic differential equations describe dynamics under uncertainty,
but their nonlinear coefficients and noise averaging complicate quantum
representations. We develop a Kolmogorov--Lindblad mapping that encodes their
probability laws as the position diagonals of trace-one quantum density
operators. For each Brownian realisation, a stochastic flow transports the
initial ensemble; the square-root Jacobian makes its action on half-densities
unitary. Averaging the resulting pure-state projectors gives a Lindblad
equation with Hermitian jump operators. Its diagonal reproduces the
Fokker--Planck density, while forward and backward intertwining identities
recover bounded observables and time correlations independently of the
initial coherences. A Galerkin approximation obtained by projecting the
Stratonovich generators preserves the Lindblad structure. We give a
residual-based error estimate and conditional quantum costs that display
the dimension dependence of approximation constants, operator
normalisations, state preparation and readout. Numerical experiments for
double-well Langevin dynamics and noisy Lorenz--63 show rapid convergence
of selected statistics at fixed dimension. The construction provides an
exact bridge from flow-regular nonlinear diffusions to quantum channels;
any computational advantage additionally requires controlled approximation
and coherent access for the chosen problem family.
\end{abstract}

\maketitle

\section{Introduction}
\label{sec:introduction}

Stochastic differential equations (SDEs) describe dynamics under uncertainty
in molecular simulation, turbulent transport, filtering, finance, and
stochastic optimisation
\cite{Pavliotis2014Stochastic,van1992stochastic,pope2000turbulent,
sarkka2019applied,black1973pricing,ermoliev1988numerical}.
The quantities of interest are often expectations, event probabilities, or
time correlations rather than individual trajectories.  Computing these
outputs requires either sampling trajectories or approximating their
probability law \cite{kloeden1992numerical,giles2008multilevel}.
Quantum algorithms offer alternative routes to statistical estimation
\cite{Montanaro2015QuantumMonteCarlo,AnLindenLiuMontanaroShaoWang2021},
but their usefulness depends on a coherent representation of both the
dynamics and the requested observable.

For nonlinear SDEs, the linearity of quantum mechanics makes this
representation a substantive issue.  The Fokker--Planck equation provides a
linear evolution of the probability density even when the trajectory
dynamics is nonlinear.  However, probabilities are normalised in $L^1$,
whereas quantum amplitudes are normalised in $L^2$.  If density values
$p_k$ and positive quadrature weights $w_k$ satisfy $\sum_k w_kp_k=1$,
direct amplitude encoding gives
\begin{equation}
    \ket{p}_{\mathrm{amp}}
    =\frac{\sum_k\sqrt{w_k}\,p_k\ket{k}}
    {\left(\sum_k w_kp_k^2\right)^{1/2}}.
    \label{eq:direct-density-encoding}
\end{equation}
Its measurement probabilities are proportional to $w_kp_k^2$; recovering
linear statistics also requires the generally unknown $L^2$ normalisation.
This issue must be addressed in methods that encode density values directly
\cite{TennieMagri2025FokkerPlanck,JinLiuYu2026FokkerPlanck,
gnanasekaran2023efficient}.  The half-density encoding
\begin{equation}
    \ket{\sqrt p}=\sum_k\sqrt{w_kp_k}\ket{k}
    \label{eq:half-density-encoding-intro}
\end{equation}
instead matches probability normalisation and statistical readout.
Its difficulty is dynamical: $\sqrt{p(t,\cdot)}$ does not generally obey a
closed linear equation.  We resolve this by lifting the density transported
by each fixed Brownian realisation before averaging over the noise.

Consider the It\^o SDE
\begin{equation}
    \dd X_t=b(X_t)\dd t+\sigma(X_t)\dd W_t,
    \label{general-SDEs}
\end{equation}
where $W_t$ is $m$-dimensional Brownian motion.  Under the flow and domain
assumptions of Section~\ref{sec:KL}, applying the same Brownian signal to every
initial point produces a random diffeomorphism of state space
\cite{Kunita1990StochasticFlows}.  A transported density carries the inverse
Jacobian determinant.  Its square root carries the square-root determinant,
which preserves the $L^2$ inner product.  Invertibility of the flow then
makes the induced half-density map unitary.  Fixing the noise still leaves
the entire initial ensemble, rather than selecting one sample trajectory.

The pathwise half-density $\psi_t^\omega$ satisfies a linear stochastic
Schr\"odinger equation.  Its generators are
\[
 H_j=-\ii\left(V_j\cdot\nabla+\tfrac12\nabla\cdot V_j\right),
\]
where $V_j=\sigma_j$ for $j\geq1$ and
$V_0=b-\tfrac12\sum_j(D\sigma_j)\sigma_j$ is the Stratonovich drift.
They have self-adjoint realisations under the standing assumptions.
Averaging the pure states,
$\Gamma(t)=\E_\omega[\ket{\psi_t^\omega}\bra{\psi_t^\omega}]$,
gives the Lindblad equation
\begin{equation}
    \frac{\dd\Gamma}{\dd t}
    =-\ii[H_0,\Gamma]
     -\frac12\sum_{j=1}^{m}[H_j,[H_j,\Gamma]].
    \label{eq:KL-intro}
\end{equation}
We call this the \emph{Kolmogorov--Lindblad mapping} (KLM).
Its generally mixed state represents the classical law through
\begin{equation}
    p(t,x)=\Gamma(t;x,x),\qquad
    \E[\phi(X_T)]=\Tr(M_\phi\Gamma(T)),
    \label{eq:KL-readout-intro}
\end{equation}
where $M_\phi$ multiplies by the classical payoff $\phi$.
The state need not equal the projector onto $\sqrt{p(t)}$.
Forward and backward semigroup identities make the readout independent of
initial coherences and also give a regression formula for time correlations.

The geometric construction has established antecedents.  Applebaum
developed the intrinsic half-density space, its half-divergence generator,
and unitary Brownian transport along a complete vector field, including
conjugation of multiplication observables \cite{Applebaum1992Operator}.
Deterministic Koopman--von Neumann methods likewise represent classical
transport by wavefunctions \cite{joseph2020Koopman}; density-operator
formulations and quantum embeddings appear in the Koopman program of
Giannakis and collaborators
\cite{Giannakis2019QuantumDA,GiannakisEtAl2022Embedding}.
KLM combines these structures with the original SDE noise to obtain a
random-unitary channel for continuously forced nonlinear dynamics.

Other quantum linearizations encode different classical objects.  Carleman
methods retain polynomial hierarchies for nonlinear ODEs and PDEs
\cite{liu2021efficient,wu2025quantum,jennings2025quantum,an2026quantum},
including stochastic extensions for stable quadratic systems
\cite{LiCatliLimPocrnicAnLiuWiebe2026NSDE}.
Bosonic and KvN constructions also support dissipative evolution and
molecular simulation
\cite{gan2025provably,watanabe2026end,
WatanabeNishiKosugiHidakaSakuraiMatsushita2026KvNGreenKubo}.
Bravyi \emph{et al.} encode a Gaussian-weighted backward Kolmogorov equation
in an interacting-oscillator space
\cite{bravyi2025quantum,bravyi2026quantum}; their structured quadratic-drift
assumptions yield polylogarithmic dimension dependence for selected
low-order correlations.  Our earlier dilation of linear It\^o SDEs instead
encoded stochastic trajectories on an auxiliary coordinate
\cite{WuLi2026LinearSDE}.

A closely related construction is due to Leng \emph{et al.}
\cite{LengDingChenLin2026OperatorSampling}.  Their operator-level Gibbs
sampler uses a factorisation of the Witten Laplacian; its warm-start
construction also gives a Lindblad representation of the transient
overdamped Fokker--Planck law.  KLM develops a stochastic-flow realisation
for general drift and multiplicative diffusion, with Hermitian jumps and a
structure-preserving approximation interface.

Our contributions are the following combined structural and computational
development.
\begin{enumerate}[label=(\arabic*),leftmargin=*,itemsep=2pt,topsep=2pt]
\item We establish the unitary half-density evolution and its averaged
Lindblad representation under explicit flow and domain assumptions.
Forward and backward Kolmogorov identities reproduce the law for every
initial encoding with the prescribed diagonal and yield bounded-observable
and two-time correlation formulas (Section~\ref{sec:KL}).
\item We compress the Hermitian Stratonovich generators before forming
the It\^o correction, preserving complete positivity and trace at finite
Galerkin dimension.  We connect approximation errors and projected
generator normalisations to quantum simulation and observable-estimation
costs (Section~\ref{sec:implementation}).
\item We compare the projected Lindblad evolution with independent
Fokker--Planck calculations for two double-well Langevin ensembles and noisy
Lorenz--63, measuring density and observable errors against the Hermite
cutoff (Section~\ref{sec:numerical-experiments}).
\end{enumerate}

The continuum representation is exact; an efficient quantum algorithm
additionally requires controlled approximation constants, state
preparation, generator access, and readout.  In particular,
Section~\ref{subsec:potential-advantage} exposes the dimension-dependent factors
in the Galerkin error and Lindblad simulation normalisation, and uses them
to compare the query bound with classical computational costs.  The
present results do not establish a generic quantum speedup.
Section~\ref{sec:conclusions} identifies concrete comparisons and analytical
questions needed to establish an advantage for particular SDE families.

\section{Kolmogorov--Lindblad mapping}
\label{sec:KL}

We now develop the construction in a form intended to be transparent rather
than maximally general.  We work on \(\R^d\) as the state space.  Periodic domains can be treated
analogously.  The formulation proceeds in three steps, each of which trades
one description of the dynamics for another:
\begin{enumerate}[label=(\roman*),leftmargin=*,itemsep=2pt,topsep=2pt]
\item \emph{Trajectories to pathwise densities.}  For each fixed Brownian
realisation, the nonlinear trajectory dynamics is replaced by a linear
stochastic transport equation for a random Eulerian density
(Section~\ref{subsec:sde-fpe}).
\item \emph{Densities to half-densities.}  The square root of the pathwise
density obeys a linear SSE whose generators are Hermitian KvN-type operators.
The pathwise evolution is unitary under the stochastic flow formulation 
(Section~\ref{subsec:half-density-lindblad}).
\item \emph{Averaging to Lindblad.}  Averaging the rank-one projectors (pure states) over
the Brownian ensemble produces a Lindblad equation whose position-space
diagonal reproduces the Fokker--Planck density
(Section~\ref{subsec:half-density-lindblad}).
\end{enumerate}
Section~\ref{subsec:structure} then establishes the structural properties of the overall
representation.  These include forward and backward intertwining identities,
the phase and mixed-state freedom, the random-unitary channel structure,
two-time correlations, and the deterministic KvN limit.
Appendix~\ref{subsec:overdamped-example} gives the overdamped Langevin example.

We assume that the real vector fields \(V_0,V_1,\ldots,V_m\) are smooth
and that \eqref{eq:stratonovich-sde} generates a jointly measurable global
stochastic flow \(\Phi_{t,s}^\omega\) of spatially smooth
diffeomorphisms on \(\R^d\), continuous in time together with its spatial
derivatives on compact sets.  In particular, for almost every realisation
the map and its inverse exist on the whole state space over every finite
time interval.  We also assume that each deterministic vector-field flow
exists for all positive and negative times.  This latter property is
\emph{two-sided completeness}.  Globally Lipschitz vector fields with at
most linear growth and bounded derivatives of positive order give one
convenient sufficient set of hypotheses
\cite{Kunita1990StochasticFlows}.  They are not necessary: the polynomial
drifts used in our numerical examples satisfy the flow requirement by the
energy estimates in Appendix~\ref{app:example-flow-completeness}.

The construction of the quantum channel below uses the flow directly and
is valid for all trace-class initial operators.  Differential calculations
use \(\mathcal C=C_c^\infty(\R^d)\) and the finite-rank test operators
spanned by \(\ket{u}\!\bra{v}\), \(u,v\in\mathcal C\).  These test
spaces are preserved by the differential expressions; invariance under
the averaged semigroup is not required.  Strong stochastic or operator
differential equations are asserted only for data in the relevant
operator domains with the required integrability.  For general data we
use the weak formulation given below.  This distinction matters because
the continuum Hamiltonians are unbounded.  Integration by parts beyond
compactly supported test functions additionally requires vanishing
far-field boundary terms.

\subsection{From SDEs to pathwise densities}
\label{subsec:sde-fpe}

Consider  It\^o SDEs
\begin{equation}
    \dd X_t
    =
    b(X_t)\,\dd t
    +
    \sigma(X_t)\,\dd W_t,
    \qquad X_t\in\R^d,
    \label{eq:ito-sde}
\end{equation}
where \(b:\R^d\to\R^d\),
\(\sigma:\R^d\to\R^{d\times m}\), and \(W_t\) is an
\(m\)-dimensional Brownian motion.  Let
\(a(x)=\sigma(x)\sigma(x)^\dag\).  The evolution equation for
probability densities originates in the work of Fokker and Planck
\cite{Fokker1914,Planck1917} and was formulated for general Markov diffusion
processes by Kolmogorov \cite{Kolmogorov1931}.  If the law of \(X_t\) has
density \(p(t,x)\), then \(p\) satisfies the forward Kolmogorov, or
Fokker--Planck, equation \cite{Pavliotis2014Stochastic}
\begin{equation}
    \partial_t p
    =
    -\nabla\cdot(bp)
    +\frac12\sum_{k,\ell=1}^{d}
    \partial_k\partial_\ell(a_{k\ell}p).
    \label{eq:fokker-planck-ito}
\end{equation}
The density \(p(t,x)\) is deterministic and describes the probability law of
the random variable \(X_t\).

For the half-density construction, it is useful to rewrite
\eqref{eq:ito-sde} in Stratonovich form,
\begin{equation}
    \dd X_t
    =
    V_0(X_t)\,\dd t
    +
    \sum_{j=1}^{m}V_j(X_t)\circ\dd W_t^j,
    \label{eq:stratonovich-sde}
\end{equation}
where \(V_j=\sigma_j\) for \(j\geq1\), with \(\sigma_j\) the
\(j\)th column of \(\sigma\), and
\begin{equation}
    V_0
    =
    b-\frac12\sum_{j=1}^{m}(\nabla\sigma_j)\sigma_j.
    \label{eq:stratonovich-drift}
\end{equation}
The Stratonovich form is the natural starting point for the geometric
construction.  Stratonovich calculus obeys the ordinary chain rule, so the
noise vector fields \(V_1,\ldots,V_m\) transport densities and
half-densities in the same way as the drift.  If
\begin{equation}\label{TV}
    \mathcal T_Vp=-\nabla\cdot(Vp)
\end{equation}
denotes the transport, or Liouville, operator of a vector field \(V\), then
the Fokker--Planck equation \eqref{eq:fokker-planck-ito} takes the Stratonovich sum-of-squares form
\begin{equation}
    \partial_t p
    =
    \mathcal T_{V_0}p
    +\frac12\sum_{j=1}^{m}\mathcal T_{V_j}^{2}p.
    \label{eq:fokker-planck-stratonovich}
\end{equation}

We now introduce the \emph{pathwise} density, which should be distinguished from
the deterministic Fokker--Planck density.  For a fixed Brownian realisation labelled by
\(\omega\), let
\(\Phi_{t,s}^\omega:\R^d\to\R^d\) be the solution map from time \(s\)
to time \(t\), and write \(\Phi_t^\omega=\Phi_{t,0}^\omega\).  More
explicitly, for every initial point \(x\in\R^d\),
\(\Phi_t^\omega(x)\) solves
\begin{equation}
    \dd\Phi_t^\omega(x)
    =
    V_0\bigl(\Phi_t^\omega(x)\bigr)\,\dd t
    +
    \sum_{j=1}^{m}
    V_j\bigl(\Phi_t^\omega(x)\bigr)
    \circ\dd W_t^j(\omega)
    \label{eq:stochastic-flow}
\end{equation}
with $    \Phi_0^\omega(x)
    =x
$; see \cite{Kunita1990StochasticFlows}. 
Importantly, the same Brownian path \(W_t(\omega)\) is used for every starting
point \(x\).  Thus \(\Phi_t^\omega\) is one random configuration of the
entire state space, not a collection of independently forced trajectories
at different spatial points.  Under the standing assumptions,
\(x\mapsto\Phi_{t,s}^\omega(x)\) is almost surely a \(C^1\)
diffeomorphism and the maps obey
$ \Phi_{t,r}^\omega\circ\Phi_{r,s}^\omega
    =
    \Phi_{t,s}^\omega$. See
\cite{Kunita1990StochasticFlows} for general discussions.

Now let $X_0$ from Equations~\eqref{eq:ito-sde} and~\eqref{eq:stratonovich-sde} have density
$p_0$ and be independent of the Brownian motion.  For a fixed Brownian
realisation $\omega$, the stochastic flow maps every possible initial point
to
\begin{equation*}
    X_t=\Phi_t^\omega(X_0).
\end{equation*}
Fixing the driving noise does not select a single trajectory, since $X_0$
remains distributed according to $p_0$.  The resulting noise-resolved
ensemble density is the pushforward
\begin{equation}
    p_t^\omega(x)
    =
    p_0\!\left((\Phi_t^\omega)^{-1}(x)\right)
    \left|
       \det D(\Phi_t^\omega)^{-1}(x)
    \right|.
    \label{p-omega}
\end{equation}
Thus $p_t^\omega$ describes the entire initial ensemble transported by one
common noise realisation.  It satisfies the stochastic continuity equation
\begin{equation}
    \dd p_t^\omega
    =
    -\nabla\cdot(V_0p_t^\omega)\,\dd t
    -
    \sum_{j=1}^{m}
    \nabla\cdot(V_jp_t^\omega)\circ\dd W_t^j.
    \label{eq:stochastic-continuity}
\end{equation}
This is the stochastic Liouville equation associated with the flow
\cite{Kunita1990StochasticFlows}.  A compact derivation is given in
Appendix~\ref{app:pathwise}.  Averaging $p_t^\omega$ over the Brownian realisations
then recovers the full Fokker--Planck ensemble.
\begin{equation}
    p(t,x)=\E_\omega\left[p_t^\omega(x)\right].
    \label{eq:fpe-density-average}
\end{equation}
Indeed, converting \eqref{eq:stochastic-continuity} to It\^o form gives
\begin{equation*}
    \dd p_t^\omega
    =
    \left(
        \mathcal T_{V_0}
        +
        \frac12\sum_{j=1}^{m}\mathcal T_{V_j}^{2}
    \right)p_t^\omega\,\dd t
    +
    \sum_{j=1}^{m}\mathcal T_{V_j}p_t^\omega\,\dd W_t^j.
\end{equation*}
Taking expectation gives the Fokker--Planck equation in \eqref{eq:fokker-planck-stratonovich}.  At this
point the nonlinearity of the SDE has already disappeared.  For each noise
realisation, \eqref{eq:stochastic-continuity} is a \emph{linear} first-order
transport equation.  The nonlinear coefficients enter only through the
vector fields.  What remains is to map \eqref{eq:stochastic-continuity} to a quantum dynamics.

\subsection{Stochastic half-densities and the Lindblad equation}
\label{subsec:half-density-lindblad}

The flow provides a direct way to see both the half-density and its
unitarity.  A small region of initial conditions is carried to a new region
by \(\Phi_t^\omega\).  Its probability mass is unchanged, but its volume
can expand or contract, so its density acquires an inverse Jacobian.
The corresponding wavefunction must acquire the square root of that
factor for its squared modulus to remain the transported density.
This geometric construction avoids differentiating \(\sqrt p\) at zeros
of \(p\).

Let
\(p_0\geq0\) be normalised in \(L^1(\R^d)\).  Choose a real phase
\(S_0\), measurable for the \(L^2\) construction, and set
\begin{equation}
    \psi_0(x)
    =
    e^{\ii S_0(x)}\sqrt{p_0(x)}.
    \label{eq:phase-choice}
\end{equation}
For a fixed realisation of the stochastic flow, define from \eqref{p-omega}
\begin{equation}
    \psi_t(x)
    =
    \psi_0\!\left((\Phi_t^\omega)^{-1}(x)\right)
    \left|
      \det D(\Phi_t^\omega)^{-1}(x)
    \right|^{1/2}.
    \label{eq:half-density-pullback}
\end{equation}
We suppress the superscript \(\omega\) on \(\psi_t\) when no confusion
can arise.  This object is a random wavefunction on the state space, not the
stochastic trajectory \(X_t\).  Comparing \eqref{eq:half-density-pullback}
with the pushforward formula for \(p_t^\omega\) gives
\begin{equation*}
    |\psi_t(x)|^2=p_t^\omega(x)
\end{equation*}
pathwise, as an equality of \(L^1\) functions.  The same formula defines
a linear operator \(U_t^\omega\) on \emph{every} \(L^2\) function,
not just on the chosen initial half-density.  Change of variables preserves
all inner products, and the explicit inverse is
\begin{equation}
    ((U_t^\omega)^{-1}g)(y)
    =\left|\det D\Phi_t^\omega(y)\right|^{1/2}
       g(\Phi_t^\omega(y)).
    \label{eq:half-density-inverse}
\end{equation}
Consequently \(U_t^\omega\) is onto and unitary, and
\(\|\psi_t\|_2=\|\psi_0\|_2=1\).  Preserving the norm of one solution
alone would not establish this conclusion.  Nor does unitarity require
volume-preserving trajectories: the square-root Jacobian compensates for
compression, while the inverse flow permits recovery of the initial
wavefunction for each fixed noise realisation.  Averaging over an unknown
noise realisation can still produce irreversible diffusion.

The square-root Jacobian is the origin of the term ``half-density'' and of
the unitary transport law.  This stochastic operator formulation goes back
to Applebaum \cite{Applebaum1992Operator}, who constructed the
Hilbert-space action on half-densities, its infinitesimal generator, and
unitary Brownian evolution along a complete vector field.
Applebaum's purpose was the operator representation of stochastic
flows.  Here the half-density evolution is used to obtain a law-level
quantum dynamics, an exact Fokker--Planck intertwining identity, and a
quantum-observable interface for general multidirectional diffusions.

For a real vector field \(V\), define
\begin{equation}
    K_V\psi
    =
    V\cdot\nabla\psi
    +
    \frac12(\nabla\cdot V)\psi,
    \label{eq:half-density-generator}
\end{equation}
so that \(-K_V\) is the infinitesimal generator of the transported
half-density.  The divergence term is exactly the weight needed to convert the
pushforward of a density into an \(L^2\)-isometric action on its square
root.  With our sign convention, \(V\cdot\nabla\) generates the forward
Koopman evolution of observables \(f\mapsto f\circ\Phi_t\), whereas a
transported half-density satisfies \(\partial_t\varphi=-K_V\varphi\).  

The role of the divergence term is especially transparent for
incompressible dynamics.  If \(\nabla\cdot V=0\), the flow preserves volume,
\(K_V=V\cdot\nabla\), and no Jacobian-dependent amplitude correction is
needed.  More generally, a vector field is incompressible relative to a
density \(\pi\) when \(\nabla\cdot(\pi V)=0\).  In underdamped Langevin
dynamics, the Hamiltonian field
\begin{equation*}
    V_{\mathrm{Ham}}(q,v)
    =
    \bigl(v,-\nabla U(q)\bigr)
\end{equation*}
is divergence free and obeys
\(\nabla\cdot(\pi V_{\mathrm{Ham}})=0\) for the Gibbs density as well.

For smooth data, differentiating the flow identity
\(\psi_t(\Phi_t^\omega(y))\,|\det D\Phi_t^\omega(y)|^{1/2}
=\psi_0(y)\) gives
\begin{equation}
    \dd\psi_t
    =
    -K_{V_0}\psi_t\,\dd t
    -
    \sum_{j=1}^{m}K_{V_j}\psi_t\circ\dd W_t^j.
    \label{eq:half-density-sde}
\end{equation}
The Brownian motions in this equation are exactly the same scalar processes
as those in the original SDE.  Each increment acts simultaneously on the
entire spatial wavefunction.   Applying the ordinary Stratonovich product rule gives
\begin{equation*}
    \dd|\psi_t|^2
    =
    2\operatorname{Re}(\overline{\psi_t}\,\dd\psi_t).
\end{equation*}
Substitution of \eqref{eq:half-density-sde} recovers
\eqref{eq:stochastic-continuity}.  Because the \(K_{V_j}\) differentiate
the spatial variable, \eqref{eq:half-density-sde} is an SPDE
\cite{DaPratoZabczyk2014SPDE}.  The randomness is nevertheless finite
dimensional: each
\(W_t^j\) is a scalar Brownian motion in time whose increment acts on the
whole wavefunction;  there is no Brownian sheet nor spatial
white noise.

The representation \(p_t^\omega=|\psi_t|^2\) has a phase freedom.  The
initial phase \(S_0\) in \eqref{eq:phase-choice} is transported by the same
inverse stochastic flow and does not affect the diagonal density.  The same
phase freedom occurs in the Koopman--von Neumann representation of
nonlinear deterministic dynamics.  Joseph used this representation for
unitary quantum simulation and, together with later Koopman--van Hove
developments, studied additional phase structure connected with canonical
classical dynamics \cite{joseph2020Koopman,TronciJoseph2021Koopman}.  KLM
does not require that additional structure.  The phase is a gauge freedom
for the classical law and multiplication observables, but not for arbitrary
nonlocal quantum observables.  Whether a particular phase can improve state
preparation or spatial approximation is an interesting question that we
leave open.

To express \eqref{eq:half-density-sde} as a stochastic Schr\"odinger
equation (SSE)~\cite{breuer2002theory}, we identify the self-adjoint
Hamiltonian associated with each complete vector-field flow.

\begin{lemma}[Half-density transport and Hermiticity]
\label{lem:skew}
Let $V$ be a smooth real vector field whose deterministic flow
$\varphi_t$ exists for every $t\in\R$.  Define
\begin{equation*}
    (U_V(t)\psi)(x)
    =
    \left|\det D\varphi_t^{-1}(x)\right|^{1/2}
    \psi\!\left(\varphi_t^{-1}(x)\right).
\end{equation*}
Then $U_V(t)$ is a strongly continuous unitary group on
$L^2(\R^d)$.  For every $\psi\in C_c^\infty(\R^d)$,
\begin{equation*}
    \left.\frac{\dd}{\dd t}U_V(t)\psi\right|_{t=0}
    =
    -K_V\psi .
\end{equation*}
Thus $K_V$ is skew-symmetric on $C_c^\infty(\R^d)$, and Stone's
theorem supplies a self-adjoint operator $H_V$ such that
$U_V(t)=e^{-\ii tH_V}$ and $H_V\psi=-\ii K_V\psi$ on this test space.
\end{lemma}

\begin{proof}
The change-of-variables formula gives
\begin{equation*}
    \|U_V(t)\psi\|_2^2
    =
    \int_{\R^d}|\psi(z)|^2\,\dd z
    =
    \|\psi\|_2^2.
\end{equation*}
Completeness makes \(\varphi_t\) a diffeomorphism with inverse
\(\varphi_{-t}\).  Hence \(U_V(t)\) is onto and therefore unitary.
Continuity of the flow and Jacobian gives strong continuity on compactly
supported smooth functions, and density and unitarity extend it to all of
\(L^2\).  Differentiating at \(t=0\) gives \(-K_V\).
Equivalently, integration by
parts over a ball \(B_R\) gives
\begin{equation*}
    \ip{\varphi}{K_V\psi}_{B_R}
    +
    \ip{K_V\varphi}{\psi}_{B_R}
    =
    \int_{\partial B_R}
       (V\cdot n)\overline{\varphi}\psi\,\dd S.
\end{equation*}
The boundary term is zero on \(\mathcal C=C_c^\infty(\R^d)\), and it tends
to zero on larger domains under the stated far-field condition.  This proves
skew symmetry.  Stone's theorem identifies the self-adjoint generator of
this unitary group, whose restriction to the test space is
\(H_V=-\ii K_V\).
\end{proof}
The far-field condition has the intuitive meaning that no \(L^2\)
probability amplitude is lost through a boundary at infinity.  It proves
formal skew symmetry.  Two-sided completeness is the additional ingredient
that makes the isometry onto and supplies self-adjointness.

In light of Lemma~\ref{lem:skew}, define the KvN Hamiltonians
\begin{equation}
    H_j=-\ii K_{V_j},
    \qquad j=0,1,\ldots,m.
    \label{eq:KvN-Hamiltonians}
\end{equation}
The stochastic half-density equation becomes
\begin{equation}
    \ii\,\dd\psi_t
    =
    H_0\psi_t\,\dd t
    +
    \sum_{j=1}^{m}H_j\psi_t\circ\dd W_t^j.
    \label{eq:strat-sse}
\end{equation}
Its equivalent It\^o form is
\begin{equation}
    \dd\psi_t
    =
    \Big(
      -\ii H_0
      -
      \frac12\sum_{j=1}^{m}H_j^2
    \Big)\psi_t\dd t
    -
    \ii\sum_{j=1}^{m}H_j\psi_t\dd W_t^j.
    \label{eq:ito-sse}
\end{equation}
This is a stochastic-Hamiltonian SSE: both the drift and the noise
generate unitary transport.  The apparently dissipative It\^o correction
\(-\frac12\sum_jH_j^2\) compensates for quadratic variation; it does not
represent pathwise loss of norm.  Indeed, for sufficiently regular data,
the two drift contributions to \(\dd\|\psi_t\|_2^2\) are
\(-\sum_j\|H_j\psi_t\|_2^2\dd t\) and
\(+\sum_j\|H_j\psi_t\|_2^2\dd t\), while each stochastic contribution
has zero real part.  The flow and its inverse establish the stronger
operator unitarity without requiring this calculation for rough initial
data.  Generic non-Hermitian Lindblad jumps do not admit this particular
unitary unraveling.

To describe the ensemble directly, define the positive trace-class
operator \(\Gamma(t)=\E_\omega[\ket{\psi_t}\!\bra{\psi_t}]\).
For regular data its position-space kernel is
\begin{equation}
\begin{aligned}
    \Gamma(t;x,y)
    ={}&
    \E_\omega\!\left[
       \psi_t(x)\overline{\psi_t(y)}
    \right].
\end{aligned}
    \label{eq:density-operator-average}
\end{equation}
This state is generally not diagonal.  For example, a deterministic initial
half-density gives
\begin{equation*}
    \Gamma(0;x,y)
    =
    e^{\ii[S_0(x)-S_0(y)]}\sqrt{p_0(x)p_0(y)},
\end{equation*}
which is nonzero away from \(x=y\) whenever \(p_0\) has extended support.
Thus the representation contains off-diagonal coherence in the position
representation, even though only its diagonal is fixed by the classical
probability law.

For data satisfying the stated domain and integrability conditions,
the It\^o product rule \cite{DaPratoZabczyk2014SPDE} applied to
$\ket{\psi_t}\!\bra{\psi_t}$, followed by expectation, gives
\begin{equation}
    \frac{\dd\Gamma}{\dd t}
    =
    -\ii[H_0,\Gamma]
    +
    \sum_{j=1}^{m}
    \left(
        H_j\Gamma H_j
        -
        \frac12\{H_j^2,\Gamma\}
    \right).
    \label{eq:KL-standard}
\end{equation}
This has the Gorini--Kossakowski--Sudarshan--Lindblad (GKSL) form
\cite{gorini1976completely,lindblad1976generators}, with Hermitian jump
operators.  Its differential expression is
\begin{equation*}
    \mathcal L(\Gamma)
    =
    -\ii[H_0,\Gamma]
    -
    \frac12\sum_{j=1}^{m}
    [H_j,[H_j,\Gamma]] .
\end{equation*}
Thus \eqref{eq:KL-standard} can be written as
\begin{equation}
    \frac{\dd\Gamma}{\dd t}
    =
    -\ii[H_0,\Gamma]
    -
    \frac12
    \sum_{j=1}^{m}[H_j,[H_j,\Gamma]].
    \label{eq:KL}
\end{equation}
Because these continuum operators are unbounded, the formal GKSL
expression alone does not establish a conservative quantum semigroup.
The flow representation below supplies that semigroup constructively.
We write it as \(\mathcal E_t=e^{t\mathcal L}\), so
\(\Gamma(t)=\mathcal E_{t-s}(\Gamma(s))\) for \(t\geq s\).
For general trace-class initial data the equation means, in particular,
\begin{equation}
    \Tr(F\Gamma(t))-\Tr(F\Gamma_0)
    =\int_0^t\Tr\!\left(\mathcal L^\dagger(F)\Gamma(r)\right)\dd r,
    \label{eq:KL-weak-form}
\end{equation}
for finite-rank tests \(F\) built from \(\mathcal C\), where
\(\mathcal L^\dagger(F)=\ii[H_0,F]
-\frac12\sum_j[H_j,[H_j,F]]\).
Each test expression extends to a bounded finite-rank operator because
the test vectors belong to the domains of \(H_0\) and \(H_j^2\).
This weak identity does not require \(H_j^2\Gamma(t)\) itself to be
trace class.

\begin{theorem}[Pathwise unitarity and the KLM quantum semigroup]
\label{thm:random-unitary}
Under the standing flow assumptions, define
\begin{equation}
\begin{aligned}
    (U_{t,s}^\omega\psi)(x)
    =&
    \left|
        \det D(\Phi_{t,s}^\omega)^{-1}(x)
    \right|^{1/2}
    \psi\!\left((\Phi_{t,s}^\omega)^{-1}(x)\right),
    \\
    U_t^\omega
    =&U_{t,0}^\omega.
\end{aligned}
    \label{eq:half-density-unitary-flow}
\end{equation}
For almost every Brownian realisation, \(U_{t,s}^\omega\) is unitary on
\(L^2(\R^d)\).  For every trace-class operator \(\Gamma_0\),
\begin{equation}
    \mathcal E_t(\Gamma_0)
    =
    \E_\omega\!\left[
       U_t^\omega\Gamma_0(U_t^\omega)^\dagger
    \right]
    \label{eq:random-unitary-channel}
\end{equation}
defines a completely positive, trace-preserving channel.  The family
\((\mathcal E_t)_{t\geq0}\) is a strongly continuous semigroup on trace
class and satisfies \eqref{eq:KL-weak-form}.  On its generator domain,
the evolution has the Lindblad differential expression \eqref{eq:KL}.
\end{theorem}
The proof is in Appendix~\ref{app:random-unitary}. 
Notice that this conclusion is stronger than norm preservation of one solution.  Every
initial trace-class state is propagated by an ensemble of unitary
conjugations.  In finite dimension, including after Galerkin projection, the
channel belongs to the mixed-unitary class studied in
Refs.~\cite{Watrous2009MixedUnitary,LeeWatrous2020MixedUnitary}.

For later use, it is helpful to define the position diagonal without
assuming that every trace-class operator has a pointwise continuous kernel.
Let \(\mathcal S_1(L^2(\R^d))\) denote the trace-class operators and let
\(M_f\) be multiplication by \(f\in L^\infty(\R^d)\).  The position-density
map
\(\mathsf D:\mathcal S_1(L^2(\R^d))\to L^1(\R^d)\) is defined by
\begin{equation}
    \int_{\R^d}
       f(x)(\mathsf D\Gamma)(x)\,\dd x
    =
    \Tr(M_f\Gamma),
    \,\; \forall f\in L^\infty(\R^d).
    \label{eq:position-diagonal-map}
\end{equation}
For example, a singular-value expansion
\(\Gamma=\sum_n s_n\ket{u_n}\!\bra{v_n}\) gives the concrete formula
\(\mathsf D\Gamma=\sum_n s_nu_n\overline{v_n}\), with convergence in
\(L^1\) since \(\|u_n\overline{v_n}\|_1\leq1\).
It follows that the definition is well posed and satisfies
\(\|\mathsf D\Gamma\|_1\leq\|\Gamma\|_1\).  If \(\Gamma\) is positive,
then \(\mathsf D\Gamma\geq0\) and
\(\int\mathsf D\Gamma=\Tr\Gamma\).  When \(\Gamma\) has a continuous
integral kernel \(\gamma_\Gamma(x,y)\), this definition reduces to
\((\mathsf D\Gamma)(x)=\gamma_\Gamma(x,x)\).  In the absence of that
regularity, \(\Gamma(t;x,x)\) below denotes \(\mathsf D\Gamma(t)\),
an \(L^1\) density rather than the restriction of an arbitrary kernel to
a set of measure zero.  With this notation,
the Fokker--Planck density is the position diagonal of the KLM state,
\begin{equation}
    \mathsf D\Gamma(t)
    =p(t,\cdot),\; 
 {\rm or }  \;\,
    \Gamma(t;x,x)
    =p(t,x).
    \label{eq:density-diagonal}
\end{equation}
In particular one has,
\begin{equation*}
    \Gamma(t;x,x)
    =
    \E_\omega[|\psi_t(x)|^2]
    =
    \E_\omega[p_t^\omega(x)]
    =
    p(t,x).
\end{equation*}
Statistical quantities, which are of eventual interest, get mapped to quantum observables,
\begin{equation}
    \E[\phi(X_t)]
    =
    \Tr\!\left(M_\phi\Gamma(t)\right),
    \;
    (M_\phi f)(x)=\phi(x)f(x),
    \label{eq:weak-observable-trace}
\end{equation}
for every bounded measurable \(\phi\).

We summarize the construction in the following statement.
\begin{theorem}[Kolmogorov--Lindblad map]
\label{prop:KL}
Assume the standing flow and domain hypotheses.  Let \(p_0\) be a normalised
probability density and let
\(\psi_0=e^{\ii S_0}\sqrt{p_0}\) for a real measurable phase \(S_0\).  Let
\(\psi_t=U_t^\omega \psi_0\) be the flow solution of
\eqref{eq:strat-sse}, interpreted weakly when necessary, driven by the same Brownian motion as
\eqref{eq:stratonovich-sde}.  Then the following statements hold.
\begin{enumerate}[label=(\roman*),leftmargin=*,itemsep=2pt,topsep=2pt]
\item Pathwise, \(|\psi_t|^2=p_t^\omega\) and
\(\|\psi_t\|_{L^2}=1\).
\item The operator
\(\Gamma(t)=\E_\omega[\ket{\psi_t}\bra{\psi_t}]\) is positive, has unit
trace, and evolves under the random-unitary channel
\eqref{eq:random-unitary-channel}.  It satisfies the weak Lindblad
equation \eqref{eq:KL-weak-form}, and \eqref{eq:KL} holds strongly whenever
the initial operator is in the generator domain.
\item Its position density coincides with the Fokker--Planck density,
\(\mathsf D\Gamma(t)=p(t,\cdot)\).
\item For every bounded measurable \(\phi\),
\(\E[\phi(X_t)]=\Tr(M_\phi\Gamma(t))\).
\end{enumerate}
\end{theorem}

\begin{proof}
The pullback formula \eqref{eq:half-density-pullback} and change of variables
give the first statement.  For regular data, applying It\^o's product rule
to the rank-one projector gives \eqref{eq:KL-standard}, hence
\eqref{eq:KL}; testing against smooth finite-rank operators gives
\eqref{eq:KL-weak-form} for general data.  Positivity and
unit trace also follow directly from the random-unitary representation.
The last two statements follow from
\eqref{eq:fpe-density-average}, \eqref{eq:density-diagonal}, and the defining
identity \eqref{eq:position-diagonal-map}.  The corresponding calculations are collected in
Appendices~\ref{app:random-unitary} and~\ref{app:pathwise}.
\end{proof}

\subsection{Structural properties}
\label{subsec:structure}

The derivation above obtains the Lindblad equation \eqref{eq:KL-standard} by averaging a particular
family of stochastic pure states.  The connection with the Fokker--Planck
equation \eqref{eq:fokker-planck-ito} is more structural.  The position-space diagonal of the KLM
evolution is closed and evolves independently of all off-diagonal
coherences.  The same relation appears in the Heisenberg picture, where
classical observables remain multiplication operators.
To make the connections explicit,
let
\begin{equation}
    \Lop_{\mathrm{FP}}
    =
    \mathcal T_{V_0}
    +
    \frac12\sum_{j=1}^{m}\mathcal T_{V_j}^{2}
    \label{eq:FP-generator-structural}
\end{equation}
be the forward Fokker--Planck generator.  The backward Kolmogorov generator
is
\begin{equation*}
    \mathcal A f
    =
    V_0\cdot\nabla f
    +
    \frac12\sum_{j=1}^{m}
    (V_j\cdot\nabla)(V_j\cdot\nabla f).
\end{equation*}
We write \(\mathsf P_t^*=e^{t\Lop_{\mathrm{FP}}}\) for the forward
semigroup on densities and
\begin{equation*}
    (\mathsf P_tf)(x)
    =
    \E[f(X_t)\mid X_0=x]
\end{equation*}
for the backward Markov semigroup on observables \cite{Pavliotis2014Stochastic}.

\begin{theorem}[Forward and backward intertwining]
\label{prop:intertwine}
Under the assumptions of Theorem~\ref{thm:random-unitary}, the KLM and Kolmogorov
semigroups obey
\begin{equation}
    \mathsf D\mathcal E_t
    =
    \mathsf P_t^*\mathsf D,
    \qquad
    \mathcal E_t^\dagger(M_f)
    =
    M_{\mathsf P_tf}
    \label{eq:semigroup-intertwining}
\end{equation}
for every trace-class operator and every bounded measurable \(f\).
On smooth test data the corresponding differential identities are
\begin{equation}
    \mathsf D(\Lop\Gamma)
    =
    \Lop_{\mathrm{FP}}(\mathsf D\Gamma),
    \qquad
    \Lop^\dagger(M_f)
    =
    M_{\mathcal A f}.
    \label{eq:intertwining}
\end{equation}
The forward identity is understood distributionally whenever only weak
derivatives exist.  The backward identity is an identity of differential
expressions on smooth tests; an identity of strong generators additionally
requires membership in their domains.  In particular, no strong
continuity on all of \(L^\infty\) is asserted.
\end{theorem}

The proof is placed in Appendix~\ref{app:intertwining}.

The pathwise and law-level relations can be displayed together as
\begin{equation}
\begin{array}{ccc}
    X_0
    & \xrightarrow{\quad\Phi_t^\omega\quad}
    & X_t^\omega=\Phi_t^\omega(X_0)
    \\[2mm]
    {\scriptstyle\operatorname{Law}}\downarrow
    &&
    \downarrow{\scriptstyle\operatorname{Law}(\,\cdot\mid W^\omega)}
    \\[2mm]
    p_0
    & \xrightarrow{\quad(\Phi_t^\omega)_*\quad}
    & p_t^\omega
    \\[1mm]
    {\scriptstyle |\cdot|^2}\uparrow
    &&
    \uparrow{\scriptstyle |\cdot|^2}
    \\[-1mm]
    \psi_0
    & \xrightarrow{\quad U_t^\omega\quad}
    & \psi_t^\omega
    \\[2mm]
    \substack{\ket{\cdot}\bra{\cdot}\\[-1mm]\downarrow}
    &&
    \substack{\E_\omega[\ket{\cdot}\bra{\cdot}]\\[-1mm]\downarrow}
    \\[2mm]
    \Gamma_0
    & \xrightarrow{\quad\mathcal E_t\quad}
    & \Gamma_t
    \\[2mm]
    \mathsf D\downarrow
    &&
    \downarrow\mathsf D
    \\[2mm]
    p_0
    & \xrightarrow{\quad\mathsf P_t^*\quad}
    & p_t.
\end{array}
\label{eq:KLM-commuting-diagram}
\end{equation}
The first three rows describe the conditional law for a fixed Brownian path.
The lower square is the exact semigroup intertwining of
Theorem~\ref{prop:intertwine}.  The diagram uses the canonical pure initial encoding
\(\Gamma_0=\ket{\psi_0}\!\bra{\psi_0}\), although the lower square applies
to arbitrary trace-class initial operators.

Thus one obtains the same classical density by first evolving
\(\Gamma_0\) with KLM and then taking its position-space diagonal, or by
evolving \(p_0=\mathsf D\Gamma_0\) directly with the Fokker--Planck
semigroup.  Although KLM generally produces off-diagonal coherences, those
coherences cannot feed back into the classical law.  This closure is special
to the differential Hamiltonians produced by KLM and does not hold for an
arbitrary Lindblad equation.

The same identity immediately applies to classical observables.  For
\(a\in L^\infty(\R^d)\),
\begin{equation}
\begin{aligned}
    \Tr(M_a\Gamma_t)
    &=
    \int_{\R^d}
        a(x)(\mathsf D\Gamma_t)(x)\,\dd x\\
    &=
    \int_{\R^d}
        a(x)(\mathsf P_t^*p_0)(x)\,\dd x.
\end{aligned}
    \label{eq:observable-from-intertwining}
\end{equation}
Consequently, a Lindblad simulation can estimate bounded classical
observables without reconstructing either the full density operator or the
full Fokker--Planck density.

\begin{remark}[Law-level phase and mixed-state freedom]
\label{rem:gauge}
Let \(\Gamma_0\succeq0\) be any trace-one operator satisfying
\(\mathsf D\Gamma_0=p_0\).  Then
\begin{equation*}
    \mathsf D\mathcal E_t(\Gamma_0)=\mathsf P_t^*p_0
\end{equation*}
for every \(t\geq0\).  Hence every such initial encoding, pure or mixed,
produces the same classical law and the same multiplication observables.
Different pure-state phases in \eqref{eq:phase-choice} give different
off-diagonal kernels but identical classical statistics.  This is a gauge
freedom only at the level of the classical law.  Nonlocal or differential
quantum observables can depend on the phase and on the coherences.
\end{remark}

The backward intertwining also gives an important route to estimate \emph{ multi-time} statistics via quantum dynamics.
\begin{corollary}[KLM quantum regression formula]
\label{cor:KLM-regression}
 Let \(X_0\) have density \(p_0\), let
\(\Gamma_0\succeq0\) satisfy
\(\Tr\Gamma_0=1\) and \(\mathsf D\Gamma_0=p_0\), and write
\(\Gamma(s)=\mathcal E_s(\Gamma_0)\). For \(0\leq s\leq t\) and
bounded real functions \(\phi_1\) and \(\phi_2\),
\begin{equation}
\begin{aligned}
    \E\!\left[
        \phi_1(X_t)\phi_2(X_s)
    \right]
    &=
    \Tr\!\left[
        M_{\phi_1}
        \mathcal E_{t-s}
        \!\left(M_{\phi_2}\Gamma(s)\right)
    \right]
    \\
    &=
    \Tr\!\left[
        \mathcal E_{t-s}^{\dagger}(M_{\phi_1})
        M_{\phi_2}\Gamma(s)
    \right].
    \label{eq:KLM-regression}
\end{aligned}
\end{equation}
The correlation depends only on
\(\mathsf D\Gamma(s)=p(s,\cdot)\).  It is therefore independent of the
initial phase and of all off-diagonal coherences in \(\Gamma(s)\).
\end{corollary}

The proof is based on the Markov property. The details are provided in Appendix~\ref{app:intertwining}.

Equation~\eqref{eq:KLM-regression} has the same algebraic form as the quantum time correlation, which is often referred to as the
quantum regression theorem due to the relation to the quantum linear response formulation \cite{Lax1963Regression}.  For a generic open
quantum system, a regression formula requires assumptions about Markovianity
and system--environment correlations.  Here the identity is exact because
the KLM channel and the classical Markov semigroup are intertwined.

There is one useful qualification.  The operator
\(M_{\phi_2}\Gamma(s)\) in \eqref{eq:KLM-regression} need not be positive,
so the first line is a linear channel identity rather than the propagation
of a physical intermediate state.  If \(\phi_2\geq0\), the same correlation
has the manifestly positive form
\begin{equation*}
\begin{split}
    \E\!\left[
        \phi_1(X_t)\phi_2(X_s)
    \right]
    =
    \Tr\!\left[
        M_{\phi_1}\mathcal E_{t-s}
        \!\left(
            M_{\sqrt{\phi_2}}\Gamma(s)M_{\sqrt{\phi_2}}
        \right)
    \right].
\end{split}
\end{equation*}
For a signed \(\phi_2\), one may apply this formula separately to its
positive and negative parts and take the difference.

For the canonical pure half-density encoding, the same quantity can also be
written as the averaged pathwise time correlation
\begin{equation}
\begin{split}
    \E_\omega\!\left[
        \left\langle
            \psi_s^\omega,
            (U_{t,s}^\omega)^\dagger
            M_{\phi_1}
            U_{t,s}^\omega
            M_{\phi_2}\psi_s^\omega
        \right\rangle
    \right].
\end{split}
    \label{eq:pathwise-KLM-correlation}
\end{equation}
Here \(U_{t,s}^\omega\) is driven by the Brownian increments between
\(s\) and \(t\).  Equation~\eqref{eq:pathwise-KLM-correlation} is a
representation of the classical two-time correlation.  It should not be
interpreted as a sequence of projective quantum measurements unless a
corresponding measurement instrument is specified.

Finally, consider the deterministic limit.  When \(m=0\), the dissipator
vanishes and \eqref{eq:KL} reduces to
\begin{equation*}
    \dot\Gamma=-\ii[H_0,\Gamma].
\end{equation*}
For a pure initial state, this is the generalized deterministic
Koopman--von Neumann evolution of a half-density
\cite{Koopman1931hamiltonian,neumann1932operatorenmethode}, including the
quantum-simulation setting of Ref.~\cite{joseph2020Koopman}.  KLM may
therefore be viewed as a stochastic open-system extension of the KvN
program.  Deterministic transport produces the coherent Hamiltonian term,
Brownian forcing produces the Lindblad dissipator, and both are constructed
from the same family of vector-field operators
\(K_{V_0},K_{V_1},\ldots,K_{V_m}\).

\section{Structure-preserving discretisation and implementation}
\label{sec:implementation}

A digital quantum algorithm requires a finite-dimensional approximation of
\eqref{eq:KL}.  We project the first-order Stratonovich generators before
forming their It\^o correction.  This preserves the quantum-channel
structure at every spatial resolution.

\subsection{Galerkin discretisation}
\label{subsec:galerkin}

Let $Y_h=\operatorname{span}\{\varphi_1,\ldots,\varphi_{N_h}\}$ have an
orthonormal basis in the common generator domain, with synthesis isometry
$\Phi_h\chi=\sum_a\chi_a\varphi_a$ and orthogonal projection
$\Pi_h=\Phi_h\Phi_h^\dagger$.  Assume $\Pi_h\to I$ strongly as the
resolution increases.  The Hermitian matrices
\begin{equation}
    H_{j,h}=\Phi_h^\dagger H_j\Phi_h,\qquad j=0,\ldots,m,
    \label{eq:projected-generators}
\end{equation}
give the coefficient-space SSE
\begin{equation}
    \ii\,\dd\chi_h=H_{0,h}\chi_h\,\dd t
       +\sum_{j=1}^{m}H_{j,h}\chi_h\circ\dd W_t^j.
    \label{eq:finite-dimensional-sse}
\end{equation}
Its propagator $U_{t,h}^\omega$ is unitary for every Brownian realisation;
$\psi_h=\Phi_h\chi_h$ is the corresponding Galerkin function.
For a trace-one continuum initial state with
$\Tr(\Pi_h\Gamma_0)>0$, use
\begin{equation}
\begin{aligned}
    G_h(0)&=\frac{\Phi_h^\dagger\Gamma_0\Phi_h}
                       {\Tr(\Pi_h\Gamma_0)},\\
    G_h(t)&=\E[U_{t,h}^\omega G_h(0)(U_{t,h}^\omega)^\dagger],
    \qquad \Gamma_h=\Phi_hG_h\Phi_h^\dagger.
\end{aligned}
    \label{eq:discrete-covariance}
\end{equation}
For $\Gamma_0=\ket{\psi_0}\bra{\psi_0}$, this corresponds to
$\chi_h(0)=\Phi_h^\dagger\psi_0/\|\Pi_h\psi_0\|_2$; the omitted
projection mass and normalisation enter the spatial error.
It\^o's formula gives the exact covariance equation
\begin{equation}
    \dot G_h=\mathcal L_h(G_h)
      :=-\ii[H_{0,h},G_h]
        -\frac12\sum_{j=1}^{m}[H_{j,h},[H_{j,h},G_h]].
    \label{eq:finite-dimensional-lindblad}
\end{equation}
Thus the semidiscrete evolution is completely positive, trace preserving,
unital, and mixed-unitary.  In particular, $G_h(t)\succeq0$ and
$\Tr G_h(t)=1$.  Its reconstructed probability density is
\begin{equation}
    p_h(t,x)=\Gamma_h(t;x,x)
      =\sum_{a,b}\varphi_a(x)(G_h(t))_{ab}\overline{\varphi_b(x)}.
    \label{eq:discrete-density-reconstruction}
\end{equation}
It is nonnegative and integrates to one.  For a bounded real payoff,
\begin{equation}
    O_{\phi,h}=\Phi_h^\dagger M_\phi\Phi_h,
    \qquad
    \Tr(O_{\phi,h}G_h)=\int\phi(x)p_h(t,x)\,\dd x,
    \label{eq:discrete-observable-readout}
\end{equation}
with $\|O_{\phi,h}\|\leq\|\phi\|_\infty$.
The consistency and observable-error estimates are given in
Appendix~\ref{app:galerkin-error}.  Nonorthogonal bases are treated by the
mass-matrix isometry $\widehat\Phi_h=\Phi_h^{\rm raw}M_h^{-1/2}$;
the corresponding conditioning and assembly details are in
Appendix~\ref{supp:implementation}.

\subsection{Block-encoded Lindblad simulation and observable estimation}
\label{subsec:backends}

Assume coherent block encodings of $H_{j,h}$ with normalisations
$\alpha_{j,h}$, including controlled access, inverses, and preparation of
the jump-index superpositions \cite{gilyen2019quantum,chakraborty2018power}.
The generic Lindblad simulator uses the scale
\begin{equation}
    \alpha_{\mathcal L,h}
      :=\alpha_{0,h}+\frac12\sum_{j=1}^{m}\alpha_{j,h}^2
      \leq C_{\rm freq}(d)h^{-2},\qquad h\leq1,
    \label{eq:KLM-Lindblad-normalisation}
\end{equation}
where the inequality is an assumption on the available implementation.
We call $C_{\rm freq}$ the \emph{generator-frequency prefactor}:
$\alpha_{\mathcal L,h}$ is an inverse-time scale for the implemented
generator, with the resolution dependence separated as $h^{-2}$.
The drift contribution is linear in its normalisation, whereas the
diffusion contributions are quadratic.  Vector-field strengths and
divergences, the number of noise channels, and normalisation losses in the
block encodings therefore enter $C_{\rm freq}$.
The derivative estimates, quadrature construction, and elementary
Hamiltonian dilation supporting these assumptions are collected in
Appendix~\ref{supp:implementation}.

The simulation query cost is $\widetilde O(T\alpha_{\mathcal L,h})$,
with $\widetilde O(mT\alpha_{\mathcal L,h})$ additional elementary gates
\cite{CleveWang2017Lindblad,LiWang2023HigherOrder}.
We state the costs below in the regime $T\alpha_{\mathcal L,h}\geq1$.
For shorter normalised times, replace $T\alpha_{\mathcal L,h}$ by
$\max\{1,T\alpha_{\mathcal L,h}\}$; this also retains the initial-state
and observable queries needed for estimation.

Here $N_h=\dim Y_h$ determines storage: the system register uses
$\lceil\log_2N_h\rceil$ qubits, while an explicit coefficient vector has
$N_h$ entries and a general density matrix has $N_h^2$ entries.
The relation between storage and resolution is basis dependent:
a tensor grid typically has $N_h\asymp h^{-d}$, whereas total-degree
Hermite cutoff $K$ has $N_h=\binom{d+K}{K}$ and permits
$h\asymp(K+1)^{-1/2}$.
Let $\alpha_{O,h}\geq1$ normalise the observable block encoding.
For coherent estimation, assume a preparation oracle for a purification
of $G_h(0)$ and a reversible Stinespring simulator with its environment
retained, so that the output preparation circuit and its inverse are
available.

\begin{theorem}[Block-encoded KLM observable estimation]
\label{thm:KLM-observable-complexity}
Under the stated access assumptions and with $T\alpha_{\mathcal L,h}\geq1$,
suppose the spatial approximation error satisfies
\begin{equation}
    |\mu_\phi(T)-\Tr(O_{\phi,h}G_h(T))|
       \leq C_{\rm approx}(d,T,\phi)h^r,\qquad r>0.
    \label{eq:KLM-weak-Galerkin-rate}
\end{equation}
Here $C_{\rm approx}$ is the approximation prefactor, including the errors
from the chosen spatial truncation and discretisation.
Choose $h$ with $C_{\rm approx}h^r\leq\epsilon/3$.
Then additive error $\epsilon$ and success probability $1-\delta$ can be
achieved with
\begin{equation}
    \widetilde O\!\left[
       \frac{\alpha_{O,h}T\alpha_{\mathcal L,h}}{\epsilon}
       \log\frac1\delta\right]
    \label{eq:KLM-observable-general-cost}
\end{equation}
primitive queries.  With
$h\asymp\min\{1,(\epsilon/C_{\rm approx})^{1/r}\}$, this gives the bound
\begin{equation}
    \widetilde O\!\left[
       \frac{\alpha_{O,h(\epsilon)}TC_{\rm freq}(d)}{\epsilon}
       \max\{1,(C_{\rm approx}/\epsilon)^{2/r}\}
       \log\frac1\delta\right].
    \label{eq:KLM-observable-smoothness-cost}
\end{equation}
\end{theorem}

\begin{proof}
Choose spatial bias below $\epsilon/3$ and combined preparation and
simulation trace-norm error below $\epsilon/(3\alpha_{O,h})$.
Trace duality bounds the latter observable error by $\epsilon/3$.
A Hadamard test for the observable block encoding, followed by coherent
amplitude estimation, estimates its normalised expectation to error
$\epsilon/(3\alpha_{O,h})$ using
$\widetilde O((\alpha_{O,h}/\epsilon)\log(1/\delta))$ output-circuit and
oracle calls.  Each output preparation uses one initial-state preparation
and $\widetilde O(T\alpha_{\mathcal L,h})$ simulation queries
\cite{LiWang2023HigherOrder}.  Since $T\alpha_{\mathcal L,h}\geq1$, this
also bounds the total primitive queries per estimation step.
The triangle inequality and substitution of
\eqref{eq:KLM-Lindblad-normalisation} and $h$ give the bounds;
precision logarithms are suppressed.
\end{proof}

These are conditional query bounds.  Gate costs, normalisation, and
spatial bias must be established for the chosen SDE, payoff, and trial
space.  The assumed approximation error must include quadrature,
far-field truncation, aliasing, and initial-projection errors;
unbounded payoffs require additional cutoff-dependent control.

\subsection{Approximation regimes and algorithmic outlook}
\label{subsec:potential-advantage}

The following discussion is exploratory: it identifies approximation
regimes that could support quantum advantage, subject to
problem-specific analysis. The actual convergence rates, prefactors,
and their dependence on dimension and evolution time will depend on
the SDE coefficients and approximation space.
Establishing these quantities requires a detailed approximation-theoretic
analysis of the chosen KLM discretisation, together with bounds on the
implemented operator normalisations and access costs. The scenarios
below should therefore be understood as conditional possibilities
whose applicability must be justified for particular problem families.

We focus on \emph{query complexity} to distinguish the effects of spatial
approximation and statistical estimation.  Calls to the preparation,
generator, and observable oracles are counted under the access assumptions
of Theorem~\ref{thm:KLM-observable-complexity}. Meanwhile, their gate-level
implementation costs are outside this comparison.  Fix $T$ and a bounded
payoff $\phi$, and abbreviate $C_{\rm approx}(d,T,\phi)$ to
$C_{\rm approx}(d)$.  We use the same approximation order $r$ as in the
theorem and the same normalised-time convention.

\paragraph{Algebraic approximation.}
Suppose the spatial bias and implemented generator scale satisfy
\begin{equation}
\begin{gathered}
    \left|\mu_\phi(T)-\Tr(O_{\phi,h}G_h(T))\right|
       \leq C_{\rm approx}(d)h^r,\\
    \alpha_{\mathcal L,h}\leq C_{\rm freq}(d)h^{-2}.
\end{gathered}
\label{eq:KLM-dimension-constants}
\end{equation}
The approximation prefactor includes projection and regularity constants,
spatial truncation, and finite-time growth; $C_{\rm freq}$ is the
block-encoding prefactor defined above.  Choosing
$h_\epsilon\asymp(\epsilon/C_{\rm approx}(d))^{1/r}$ in the resolved regime
gives
\begin{equation}
\begin{gathered}
    T\alpha_{\mathcal L,h_\epsilon}
       \lesssim T C_{\rm alg}(d)\epsilon^{-2/r},\\
    C_{\rm alg}(d):=C_{\rm freq}(d)C_{\rm approx}(d)^{2/r},
\end{gathered}
\label{eq:KLM-generic-dimension-factor}
\end{equation}
where ``alg'' denotes algebraic approximation.  Retaining the observable
normalisation makes the estimation prefactor explicit:
\begin{equation}
    C_{\rm KLM}^{\rm query}
    =\widetilde O\!\left[
       \alpha_{O,h_\epsilon}T C_{\rm alg}(d)\epsilon^{-1-2/r}
       \log\frac1\delta
    \right].
\label{eq:KLM-dimension-explicit-cost}
\end{equation}
If both approximation and frequency prefactors are polynomial in $d$ and
$r$ stays bounded away from zero, then $C_{\rm alg}(d)$ is polynomial as
well.  With bounded $\alpha_{O,h_\epsilon}$ this gives polynomial query
complexity at fixed accuracy. 

\paragraph{Sampling and observable variance.}
For comparison, ordinary Monte Carlo used in classical SDE solvers has query cost approximately
\[
    C_{\rm MC}^{\rm query}
    =\widetilde O\!\left[
       A_{\rm path}(d,\epsilon)V_\phi(d)\epsilon^{-2}
       \log\frac1\delta
    \right],
\]
in the usual sampling regime, where $V_\phi(d)$ bounds the variance of the
sampled path payoff and $A_{\rm path}$ counts coefficient and payoff queries for one path with
controlled time-discretisation bias. The vector-field data entering $C_{\rm freq}$ also determine the
cost of path generation, but $A_{\rm path}$ additionally includes
the number of time steps needed to control the bias. For example, additive-noise Langevin dynamics with bounded-degree
local interactions requires $O(d)$ local force queries per
Euler--Maruyama step. With weak error 
$C_{\rm weak}(d,T,\phi)\Delta t$ and $O(d)$ payoff cost, this gives
$A_{\rm path}=O(d\,C_{\rm weak}/\epsilon)$ at fixed $T$ in the
resolved regime, reducing to $O(d/\epsilon)$ if $C_{\rm weak}$ is
dimension independent. For KLM, encoding each of the $O(d)$ local
drift terms and noise operators with normalisation $O(h^{-1})$
uniformly in $d$ gives $C_{\rm freq}(d)=O(d)$. These weak-error
and encoding assumptions require problem-specific justification.

The continuum payoff variance has an exact quantum counterpart:
\[
    \operatorname{Var}[\phi(X_T)]
    =\Tr(M_\phi^2\Gamma(T))
       -\bigl[\Tr(M_\phi\Gamma(T))\bigr]^2,
\]
by the diagonal correspondence. 
With a reversible circuit that outputs payoff values and a known variance
bound $V$, variance-aware quantum mean estimation can use
$\widetilde O(\sqrt V/\epsilon)$ coherent circuit calls
\cite{Montanaro2015QuantumMonteCarlo}.  

For bounded observable normalisation, fixed nonzero payoff variance, and
no additional algebraic accuracy dependence in $A_{\rm path}$, the
algebraic KLM exponent $1+2/r$ improves on the ordinary Monte Carlo exponent
$2$ when $r>2$.  This possible improvement in accuracy dependence is
polynomial.

A more favourable regime arises when the spatial approximation exhibits spectral convergence,
\[
\begin{gathered}
    \left|\mu_\phi(T)-\Tr(O_{\phi,h}G_h(T))\right|
    \\
    \leq C_{\rm approx}(d)\exp[-a(d)h^{-\nu}],
    \qquad a(d)>0,\quad \nu>0.
\end{gathered}
\]
In the resolved regime, one can choose
$h_\epsilon\asymp
[a(d)/\log(C_{\rm approx}(d)/\epsilon)]^{1/\nu}$.
The resulting observable-estimation cost is
\begin{equation}
\begin{aligned}
    C_{\rm KLM}^{\rm query}
    =\widetilde O\!\Bigg[&
       \frac{\alpha_{O,h_\epsilon}TC_{\rm freq}(d)}{\epsilon}
       \log\frac1\delta\\
       &\times
       \left(
          \frac{\log(C_{\rm approx}(d)/\epsilon)}{a(d)}
       \right)^{2/\nu}
    \Bigg].
\end{aligned}
\label{eq:KLM-spectral-dimension-factor}
\end{equation}
Thus spectral convergence replaces the algebraic simulation overhead by
logarithmic factors.  At fixed dimension and bounded observable
normalisation, the total query cost is nearly proportional to
$\epsilon^{-1}$, offering a near-quadratic improvement in precision over
ordinary Monte Carlo with fixed variance and bounded path-query cost.

Importantly, spectral convergence also weakens the sensitivity to dimension-dependent
approximation constants: $C_{\rm approx}$ enters only through its logarithm.
More ambitiously, spectral convergence can permit polylogarithmic
query dependence on $d$: for fixed $\nu$ and bounded observable
normalisation, it suffices that $C_{\rm freq}(d)$,
$a(d)^{-2/\nu}$, and $\log C_{\rm approx}(d)$ grow at most
polynomially in $\log d$. For example, this holds when
$C_{\rm freq}$ is polylogarithmic, $a(d)$ is bounded below uniformly,
and $C_{\rm approx}$ grows polynomially in $d$. Establishing such
bounds for structured SDE families is an interesting direction,
motivated by the dimension dependence achieved in
Ref.~\cite{bravyi2026quantum}. 
Establishing such spectral estimates with controlled dimension dependence
is consequently a promising route to quantum advantage
\cite{IrrgeherKritzerPillichshammerWozniakowski2016Exponential}.
Weighted Hermite and hyperbolic-cross approximation, and low-rank tensor
representations, provide useful settings in which to investigate the
required approximation and oracle bounds
\cite{LuoYau2013HermiteHyperbolicCross,
LeobacherPillichshammerEbert2023Hermite,
IrrgeherKritzerPillichshammerWozniakowski2016Hermite,
ChenHoskinsKhooLindsey2023CommittorTensor,
TangYing2024FunctionalHierarchicalTensor}.

\section{Numerical experiments}
\label{sec:numerical-experiments}

We test the Galerkin--Lindblad approximation against independent
Fokker--Planck calculations for two double-well ensembles and noisy
Lorenz--63. All Lindblad equations are integrated classically: these tests
assess the representation and spatial approximation, rather than quantum
circuits or dimension-dependent speedup. The models have additive noise,
so the It\^o and Stratonovich drifts coincide. Detailed basis definitions,
operator assembly, time integration, grids and boundary diagnostics are
given in Appendix~\ref{supp:numerical-methods}.

\subsection{Underdamped Langevin dynamics in a double well}
\label{subsec:numerics-double-well}

We consider
\begin{equation}
\begin{aligned}
 \dd q_t&=v_t\dd t,\\
 \dd v_t&=-(q_t^3-q_t)\dd t-\gamma v_t\dd t
             +\sqrt{2\gamma/\beta}\,\dd W_t,
\end{aligned}
 \label{eq:numerical-double-well-sde}
\end{equation}
with potential $U(q)=(q^2-1)^2/4$, $\gamma=1$, $\beta=4$ and $T=1$.
Such models describe thermally activated transitions in molecular and
materials applications \cite{leimkuhler2015molecular,risken1996fokker,
kramers1940brownian,hanggi1990reaction}.
Writing $P_q=-\ii\partial_q$, $P_v=-\ii\partial_v$ and
$b_v=-q^3+q-\gamma v$, the KLM operators are
\begin{equation}
 H_0=\tfrac12\{v,P_q\}+\tfrac12\{b_v,P_v\},\qquad
 L=\sqrt{2\gamma/\beta}\,P_v.
 \label{eq:numerical-double-well-operators}
\end{equation}
We use a tensor Hermite space with levels $0\leq n_q,n_v\leq K$,
centres zero, scales $(0.35,0.50)$ and dimension $N=(K+1)^2$.
Polynomial products are assembled in a padded space before final
compression. The Lindblad equation is advanced by RK4 with step $10^{-3}$,
for $K=8,10,\ldots,20$. RK4 is not CPTP; the reported computations monitor
trace, Hermiticity and the smallest eigenvalue. No eigenvalue more
negative than roundoff is observed, and halving the step at $K=18$ leaves
the two reported statistics unchanged at the displayed precision.

The reference solves the Fokker--Planck equation on
$[-3,3]\times[-5,5]$ using $128^2$, $192^2$ and $256^2$ cells and time step
$10^{-3}$. A symmetric splitting combines periodic discrete diffusion with
conservative remapping of exact transport subflows. Reference observables
are extrapolated linearly in the squared grid spacing. Their reported
uncertainties are empirical spatial-refinement diagnostics, rather than
certified total errors. Boundary and escaped-mass diagnostics are small
on the finest grid, as reported in Appendix~\ref{supp:numerical-methods}.

Let $p_\pm(q,v)=\mathcal N(q;\pm1,0.25^2)\mathcal N(v;0,0.50^2)$.
Experiment I starts from $p_-$ and its pure half-density encoding.
Experiment II starts from $0.65p_-+0.35p_+$ and uses the mixed encoding
\begin{equation}
 \Gamma_0^{\rm mix}
 =0.65\ket{\sqrt{p_-}}\bra{\sqrt{p_-}}
  +0.35\ket{\sqrt{p_+}}\bra{\sqrt{p_+}}.
 \label{eq:numerical-bimodal-covariance}
\end{equation}
Both initial operators are projected and normalised as in
Section~\ref{subsec:galerkin}. The second experiment checks that propagation
does not require a rank-one initial state. The mean position and right-well
probability are evaluated through their projected multiplication operators.

\begin{table*}[t]
\centering
\caption{Double-well statistics at $T=1$. The uncertainty is the empirical
spatial uncertainty of the extrapolated Fokker--Planck reference; the last
column is the absolute Galerkin discrepancy at $K=20$ ($N=441$).}
\label{tab:double-well-statistics}
\begin{tabular}{lllll}
\hline
Ensemble & Statistic & Reference & Uncertainty & Discrepancy\\
\hline
Left well & $\mathbb P(q_T>0)$ & $0.01335670$ & $2.77\times10^{-6}$ & $1.52\times10^{-4}$\\
Left well & $\E[q_T]$ & $-0.92813490$ & $1.66\times10^{-5}$ & $5.10\times10^{-4}$\\
Mixture & $\mathbb P(q_T>0)$ & $0.35400716$ & $1.02\times10^{-6}$ & $4.54\times10^{-5}$\\
Mixture & $\E[q_T]$ & $-0.27844071$ & $5.26\times10^{-6}$ & $1.53\times10^{-4}$\\
\hline
\end{tabular}
\end{table*}

The box-restricted density $L^1$ discrepancies at $K=20$ are
$6.05\times10^{-3}$ and $5.76\times10^{-3}$ for the two ensembles.
In Experiment I the initial projection mass rises from $0.9748$ at $K=8$
to $0.999999995$ at $K=20$, while the outermost-mode population falls
from $9.08\times10^{-2}$ to $8.43\times10^{-4}$. These are useful
resolution diagnostics, although they are not error bounds.
Figures~\ref{fig:numerical-langevin-left} and~\ref{fig:numerical-langevin-bimodal}
show convergence of both statistics and agreement of the reconstructed
density with the independent reference, including the unequal populations
of the two wells.

\begin{figure*}[t]
\centering
\includegraphics[width=0.58\textwidth]{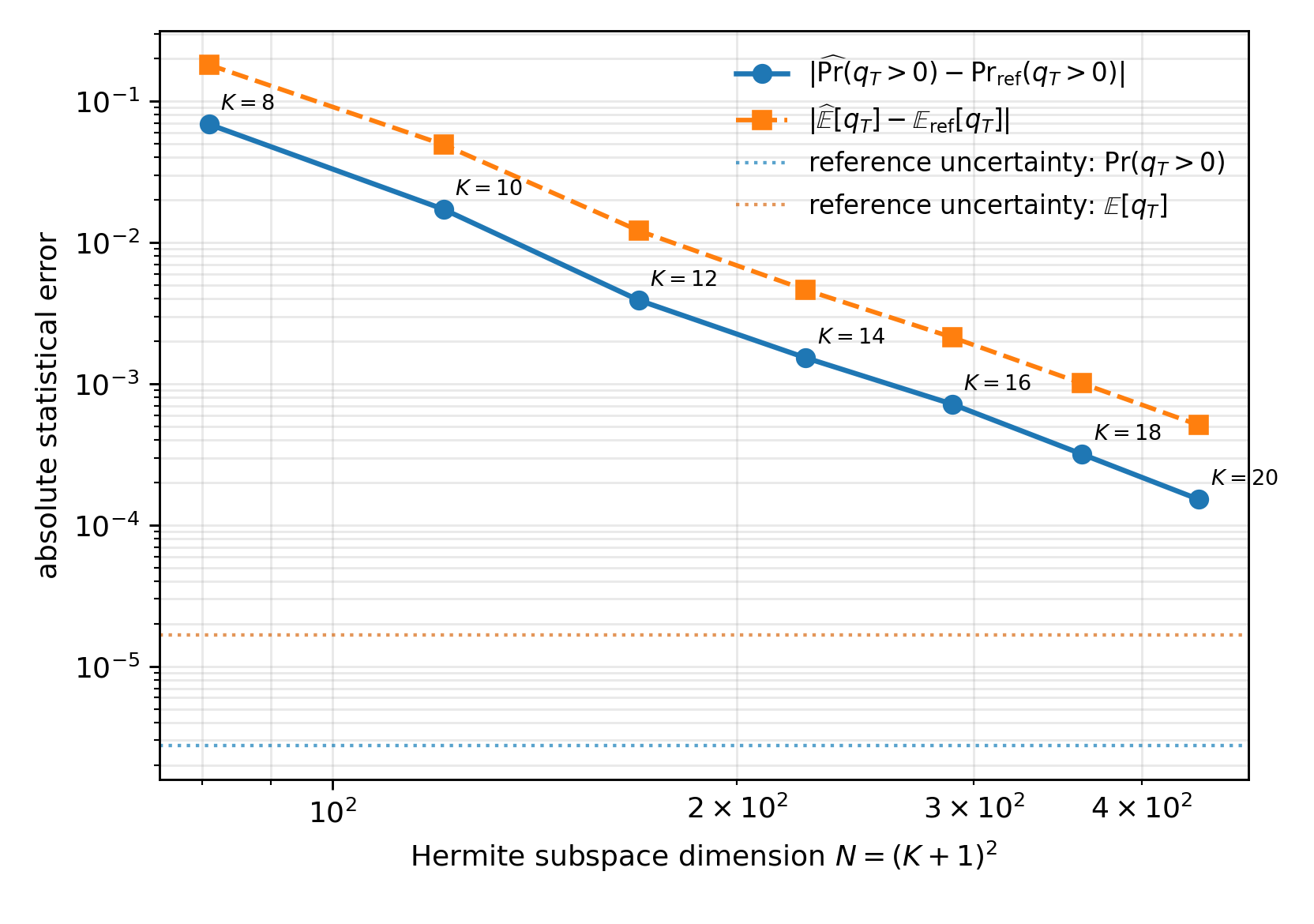}\\
\includegraphics[width=0.98\textwidth]{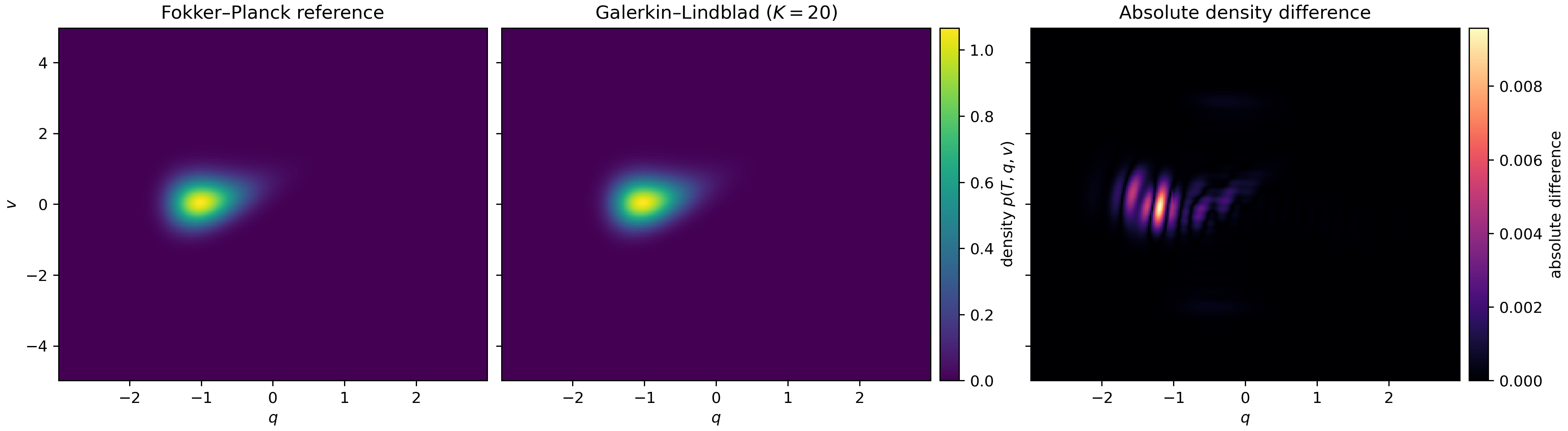}
\caption{Left-well ensemble at $T=1$. Upper: observable errors versus
$N=(K+1)^2$, with dotted spatial reference uncertainties. Lower:
$256^2$ Fokker--Planck density, $K=20$ Galerkin--Lindblad diagonal and
their absolute difference.}
\label{fig:numerical-langevin-left}
\end{figure*}

\begin{figure*}[t]
\centering
\includegraphics[width=0.58\textwidth]{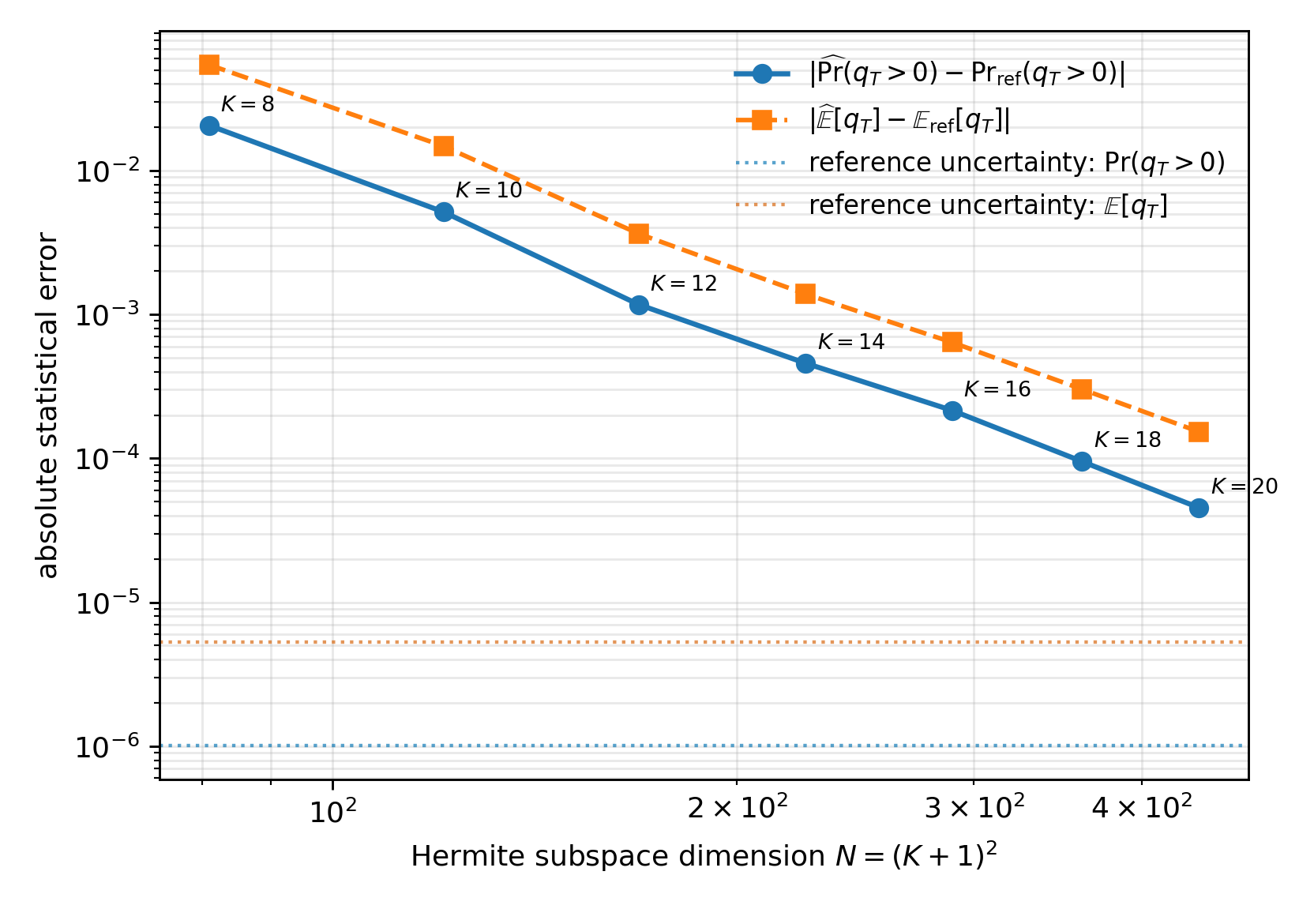}\\
\includegraphics[width=0.98\textwidth]{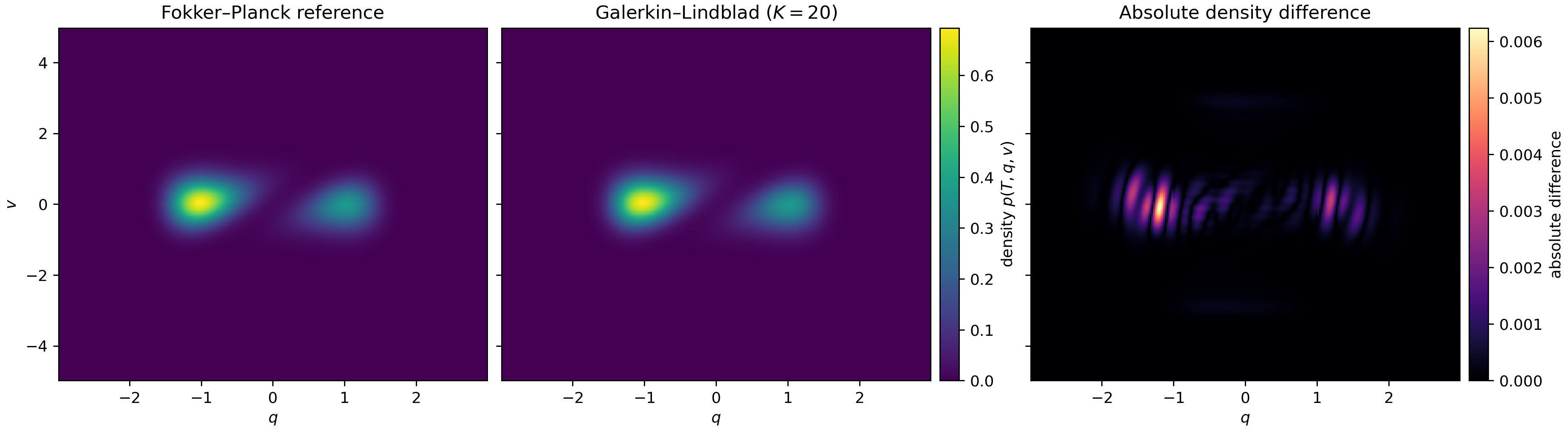}
\caption{Biased mixture at $T=1$. Upper: convergence of the right-well
probability and mean position. Lower: reference and $K=20$ Lindblad
densities and their absolute difference. Both unequal well populations
are reproduced.}
\label{fig:numerical-langevin-bimodal}
\end{figure*}

\subsection{Noisy Lorenz--63 dynamics}
\label{subsec:numerical-lorenz}

The third experiment uses
\begin{equation}
\begin{aligned}
 \dd X_t&=\sigma_L(Y_t-X_t)\dd t+\sqrt{2\kappa}\,\dd W_t^1,\\
 \dd Y_t&=[X_t(\rho_L-Z_t)-Y_t]\dd t+\sqrt{2\kappa}\,\dd W_t^2,\\
 \dd Z_t&=(X_tY_t-\beta_LZ_t)\dd t+\sqrt{2\kappa}\,\dd W_t^3,
\end{aligned}
\label{eq:numerical-lorenz-sde}
\end{equation}
with $(\sigma_L,\rho_L,\beta_L)=(10,28,8/3)$, $\kappa=0.2$ and
$T=0.05$. This standard nonlinear uncertainty-propagation benchmark
\cite{lorenz1963mechanics,sarkka2019applied} tests quadratic drift and
three diffusion channels over a short time, rather than long-time chaotic
mixing. For drift $b$, the KLM operators are
$H_0=-\ii[b\cdot\nabla-(\sigma_L+1+\beta_L)/2]$ and
$L_j=-\ii\sqrt{2\kappa}\partial_j$.

The initial law is Gaussian with mean $(-8,-8,27)$ and covariance
$1.5^2I_3$. A total-degree Hermite basis centred at this mean, with scales
$\sqrt2\,1.5$, makes $\sqrt{p_0}$ the ground state. We use $K=2,4,\ldots,14$,
$N=\binom{K+3}{3}$, and RK4 step $3.5\times10^{-4}$; halving this step
at $K=12$ changes the mean by less than $2\times10^{-10}$.
The reference uses $64^3$, $80^3$ and $96^3$ grids, a symmetric
transport--diffusion splitting and step $10^{-3}$. Halving that step on
the finest grid changes the mean by $4.5\times10^{-6}$.

The extrapolated reference mean is $-8.00826159$, with empirical spatial
uncertainty $1.04\times10^{-4}$. At $K=14$ ($N=680$), the Lindblad
mean is $-8.00796125$, giving discrepancy $3.00\times10^{-4}$.
The $X$-marginal $L^1$ discrepancy is $3.88\times10^{-3}$, its
reconstructed mass differs from one by less than $2\times10^{-10}$,
and the highest-degree-shell population decreases from $2.26\times10^{-1}$
at $K=2$ to $2.17\times10^{-4}$ at $K=14$.
Figure~\ref{fig:numerical-lorenz} shows the marginal and mean convergence.

\begin{figure*}[t]
\centering
\includegraphics[width=0.48\textwidth]{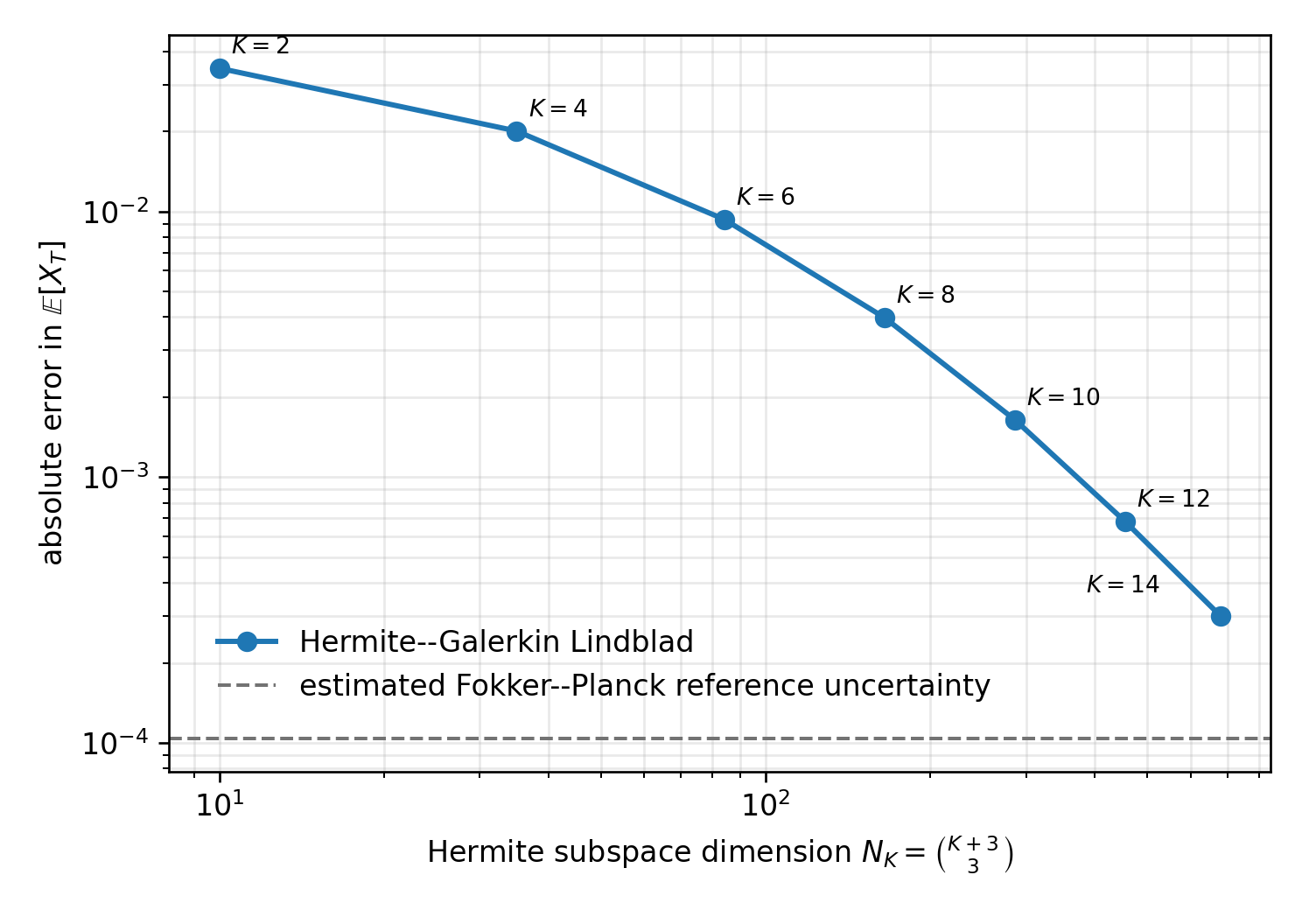}\hfill
\includegraphics[width=0.48\textwidth]{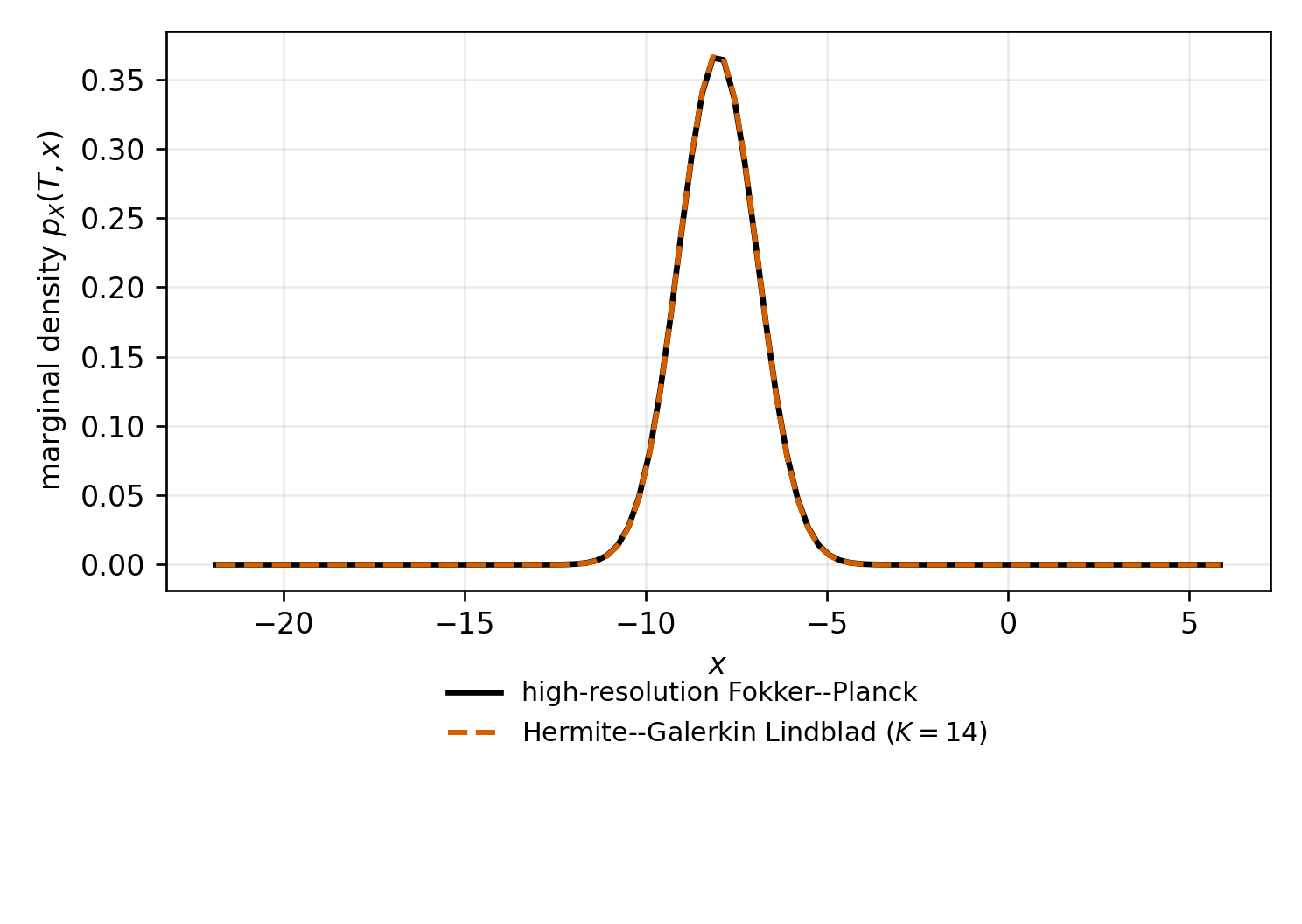}
\caption{Noisy Lorenz--63 at $T=0.05$. Left: mean error versus
$N=\binom{K+3}{3}$, with the spatial reference uncertainty. Right:
$96^3$ Fokker--Planck and $K=14$ Galerkin--Lindblad $X$-marginals.}
\label{fig:numerical-lorenz}
\end{figure*}

\subsection{Observed convergence with the Hermite cutoff}
\label{subsec:numerical-observed-complexity}

Common-slope fits for the two Langevin ensembles over
$K=12,14,\ldots,20$ give errors proportional to $e^{-0.40K}$ for the
right-well probability and $e^{-0.39K}$ for the mean. For Lorenz--63,
the fit over $K=6,8,\ldots,14$ is $0.125e^{-0.43K}$.
Each uses only five cutoffs; these are preasymptotic descriptions and
should not be extrapolated below the reference uncertainty.
If a bound $Ce^{-cK}$ held, the required dimensions would scale as
$\log^2(1/\epsilon)$ and $\log^3(1/\epsilon)$ in these fixed-dimensional
spaces. The numerical fits do not prove such bounds or their robustness
as $d$ increases.

Polynomial-degree counting gives the fixed-dimensional matrix estimates
\begin{equation}
\begin{aligned}
 \|H_{0,K}^{\rm DW}\|+\|L_K\|^2&=O(K^2),\\
 \|H_{0,K}^{\rm L63}\|+\sum_{j=1}^3\|L_{j,K}\|^2&=O(K^{3/2}).
\end{aligned}
\label{eq:numerical-projected-generator-scales}
\end{equation}
The noise contributes $O(K)$; the higher powers come from nonlinear
drift. These estimates concern the actual polynomial operators on
$\mathbb R^d$, where uniform bounded-coefficient assumptions do not apply.
The right-well observable is a positive contraction but is dense in the
Hermite basis; the coordinate observables have norms $O(\sqrt K)$.
Hence these norm estimates alone do not give efficient block encodings.
The coordinate means also require the moment control described in the
Galerkin-error analysis in Appendix~\ref{app:galerkin-error} to fall under a rigorous observable bound.
The tests validate finite-time approximation and motivate spectral-tail
analysis; they establish neither coherent-access costs nor an end-to-end
quantum advantage.

\section{Discussion and outlook}
\label{sec:conclusions}

\subsection{Structural perspective}

KLM represents the finite-time law of a nonlinear diffusion by a quantum
channel under explicit flow and domain assumptions.  The original Brownian
noise selects unitary half-density maps; averaging their conjugations gives
Hermitian-jump Lindblad evolution.  The forward and backward identities
\[
 \mathsf D\mathcal E_t=\mathsf P_t^*\mathsf D,
 \qquad \mathcal E_t^\dagger(M_f)=M_{\mathsf P_tf}
\]
identify the classical Markov semigroup with the restriction of the
Heisenberg evolution to multiplication operators.  This explains both
coherence-independent statistical readout and the regression formula.
Iterating the Markov argument gives ordered higher-time correlations of
bounded observables.  No physical quantization or semiclassical limit is
required: the quantum channel is an auxiliary representation of the given
classical process.

Compressing the Stratonovich generators before forming their quadratic
correction preserves the GKSL form.  The finite-dimensional semigroup is
therefore completely positive and trace preserving, although a chosen
classical time integrator need not retain those properties.  The numerical
tests show rapid cutoff convergence at fixed dimensions two and three.
They validate finite-time approximation for these examples, but do not
establish dimension-robust convergence or quantum advantage.

\subsection{Relation to other representations}

KLM places the classical law on the diagonal of a trace-one operator.
This avoids the density-amplitude normalisation mismatch in
\eqref{eq:direct-density-encoding}, while retaining state preparation as a
computational task.  The resulting mixed state differs from a pure
Q-sample.  Quantum-walk sampling instead starts from coherent access to
transition probabilities
\cite{Szegedy2004QuantumWalk,WocjanAbeyesinghe2008QuantumSampling,
ApersSarlette2019QFF}. In contrast, KLM constructs local generators from the SDE fields.

Leng \emph{et al.} aim to accelerate sampling from Gibbs distributions,
including non-log-concave targets, by exploiting the operator structure
of overdamped Langevin dynamics
\cite{LengDingChenLin2026OperatorSampling}.
Their sampler filters the kernel of a first-order factor $A$ of the
Witten Laplacian, $\mathsf H_W=A^\dagger A$, whose zero mode encodes
the Gibbs half-density. The singular-value gap
$\sqrt{\operatorname{gap}(\mathsf H_W)}$ enables inverse-square-root
gap dependence under their discretisation, coherent-access, and
warm-start assumptions.
Their construction also provides a direct connection to transient
dynamics: its non-Hermitian Witten-factor jumps reproduce the overdamped
Fokker--Planck law on the position diagonal while fixing the pure Gibbs
half-density (Appendix~F, Lemma~24 of Ref.~\cite{LengDingChenLin2026OperatorSampling}).
KLM targets finite-time laws and observables of general nonlinear SDEs
through a stochastic-flow construction that does not require a reversible
Gibbs model. Its Hermitian-jump, random-unitary lift reproduces the
classical law but generally has different off-diagonal evolution.
Consequently, the sampling gap improvement established by Leng
\emph{et al.}, or a corresponding reduction in discretised-generator
norm dependence, does not follow automatically for the KLM channel.

Gaussian-LCHS also exploits the reversible overdamped structure, using a
Gibbs-weighted similarity transform and sum-of-squares factorisation to
estimate reaction rates and propagator matrix elements
\cite{KharaziAlkadriMandadapuWhaley2026ReactionRates}.
The Kolmogorian method instead approximates selected backward observables
for structured quadratic SDEs
\cite{bravyi2025quantum,bravyi2026quantum}.  Its Gaussian-weighted
backward generator is intertwined with the restriction of the KLM
Heisenberg generator to multiplication operators, rather than with the
full Lindbladian.  These methods can exploit a narrower output or model
class; KLM pays for approximating a state that represents the complete
finite-time law.  Comparisons must match the requested statistic and its
normalisation, including the strongest classical trajectory, multilevel,
and compressed-density alternatives.

\subsection{Approximation and algorithmic questions}

The first analytical task is to control the approximation and generator
prefactors in the query model of Section~\ref{subsec:potential-advantage}.
For weak bias $C_{\rm approx}(d)h^r$ and available simulation
normalisation $C_{\rm freq}(d)h^{-2}$, the combined algebraic prefactor is
$C_{\rm alg}(d)=C_{\rm freq}(d)C_{\rm approx}(d)^{2/r}$.
At fixed time and success probability, and with bounded observable
normalisation, the query bound is
$\widetilde O(C_{\rm alg}(d)\epsilon^{-1-2/r})$ in the regime specified
above. The factor $\epsilon^{-2/r}$ comes from spatial resolution and
$\epsilon^{-1}$ from coherent observable estimation. Weak-error analysis
through the backward equation may improve the conservative strong-error
estimates; weighted Hermite, sparse, and low-rank spaces should be assessed
together with their generator normalisation and coherent access.

Controlled prefactors can give a substantial advantage over explicit
Galerkin evolution even with algebraic truncation error. Bravyi
\emph{et al.} illustrate this mechanism: under suitable bounds on the
coefficients and observable and readout norms, their construction admits
dimension-independent approximation prefactors and polynomial dependence
on the energy cutoff, while avoiding explicit evolution of a potentially
exponentially large Hermite space \cite{bravyi2026quantum}.
Their generator construction has a different cutoff dependence from the
$h^{-2}$ assumption used here, but likewise controls the dimension
prefactor in the implemented normalisation. Such comparisons concern
explicit coefficient representations; an advantage over classical sampling
must be assessed against the cost of estimating the same observable.

Spectral convergence offers a further opportunity. A bias of the form
$C_{\rm approx}(d)\exp[-a(d)h^{-\nu}]$ replaces the algebraic spatial
overhead by logarithmic factors and gives nearly $\epsilon^{-1}$ query
complexity when the remaining factors are controlled. For fixed $\nu$,
polynomial bounds on $C_{\rm freq}(d)$, $a(d)^{-1}$, and
$\log C_{\rm approx}(d)$ suffice for polynomial dimension dependence,
provided the observable normalisation remains bounded.

A tractable computational progression starts with stable affine diffusions
with additive noise and Gaussian initial laws, for which Gaussian
solutions and banded Hermite generators permit explicit resource
accounting, and then introduces multiplicative linear noise or weak
anharmonicity. As non-Gaussian effects appear, bounded observables and
event probabilities provide tests beyond closed moment equations.
For each model family, quantitative bounds on $C_{\rm approx}(d)$,
$C_{\rm freq}(d)$, and $\alpha_{O,h}$, together with explicit coherent access
to initial states, vector fields, and payoffs, would make the
query-complexity comparison testable.

Finally, the regression formula motivates algorithms for time correlations
\cite{li2025quantum}. For a nonnegative intermediate weight $\phi_2$,
the positive, generally unnormalised operator
$M_{\sqrt{\phi_2}}\Gamma(s)M_{\sqrt{\phi_2}}$ reproduces the required
classical correlation after propagation; signed weights can be split into
positive and negative parts. Controlling the approximation error,
preparation, normalisation, and estimation queries for these intermediate
operators would extend the present framework to transport and response
statistics.

\paragraph*{Code availability.}
The scripts used to generate the numerical data will be made publicly
available in a version-controlled repository upon publication.

\begin{acknowledgments}
The authors gratefully acknowledge support from the National Science
Foundation under Grants Nos.~DMS-2411120 and DMS-2552687, and from the ICDS
Superseed Grant at Penn State.
\end{acknowledgments}

\clearpage
\onecolumngrid
\appendix
\crefalias{section}{appendix}
\numberwithin{theorem}{section}
\numberwithin{lemma}{section}
\numberwithin{remark}{section}

\section{Proof of pathwise unitarity and random-unitary KLM evolution}
\label{app:random-unitary}

This section proves Theorem~\ref{thm:random-unitary}.  The stochastic flow
$\Phi_{t,s}^{\omega}$ introduced in \eqref{eq:stochastic-flow} transports
points in $\R^d$, whereas $U_{t,s}^{\omega}$ in
\eqref{eq:half-density-unitary-flow} is its lift to half-density functions
in $L^2(\R^d)$.  The square-root Jacobian makes this lifted action
isometric, and the inverse flow makes it onto.  Averaging the resulting
unitary conjugations then gives the random-unitary KLM channel.

\begin{proof}
Fix a Brownian realisation for which the stochastic flow and its inverse
exist.  If $f:\R^d\to\mathbb C$ is a half-density at time $s$, then
\begin{equation*}
    (U_{t,s}^{\omega}f)(x)
    =
    \left|
       \det D(\Phi_{t,s}^{\omega})^{-1}(x)
    \right|^{1/2}
    f\!\left((\Phi_{t,s}^{\omega})^{-1}(x)\right)
\end{equation*}
is another function on $\R^d$, representing the transported half-density
at time $t$.  For $s=0$ and $f=\psi_0$, this is precisely
\eqref{eq:half-density-pullback}.

For $f,g\in L^2(\R^d)$, the substitution
$y=(\Phi_{t,s}^{\omega})^{-1}(x)$ gives
\begin{align*}
    \left\langle U_{t,s}^{\omega}f,U_{t,s}^{\omega}g\right\rangle
    =
    \int_{\R^d}
    \left|\det D(\Phi_{t,s}^{\omega})^{-1}(x)\right|
    \overline{f((\Phi_{t,s}^{\omega})^{-1}(x))}
    g((\Phi_{t,s}^{\omega})^{-1}(x))\,\dd x \\
    =
    \int_{\R^d}\overline{f(y)}g(y)\,\dd y
    =
    \langle f,g\rangle .
\end{align*}
Thus $U_{t,s}^{\omega}$ is an isometry.  The half-density lift of
$(\Phi_{t,s}^{\omega})^{-1}$ supplies its inverse, so
$U_{t,s}^{\omega}$ is surjective and hence unitary.  The composition law of
the stochastic flow gives
\begin{equation*}
    U_{t,r}^{\omega}
    =
    U_{t,s}^{\omega}U_{s,r}^{\omega},
    \qquad r\leq s\leq t,
\end{equation*}
which is the stochastic cocycle property.

We next identify its differential equation.  Let
$\psi_t=U_{t,s}^{\omega}f$ and define
\begin{equation*}
    J_{t,s}^{\omega}(y)
    =
    \left|\det D_y\Phi_{t,s}^{\omega}(y)\right|.
\end{equation*}
The definition of $U_{t,s}^{\omega}$ is equivalently
\begin{equation*}
    \psi_t\!\left(\Phi_{t,s}^{\omega}(y)\right)
    \left(J_{t,s}^{\omega}(y)\right)^{1/2}
    =
    f(y).
\end{equation*}
Thus $\Phi_{t,s}^{\omega}$ describes where the point $y$ moves, while the
Jacobian describes the accompanying change in the half-density amplitude.
The flow equation \eqref{eq:stochastic-flow} and Jacobi's formula give
\begin{equation*}
    \dd\log J_{t,s}^{\omega}(y)
    =
    (\nabla\cdot V_0)(\Phi_{t,s}^{\omega}(y))\,\dd t
    +
    \sum_{j=1}^{m}
    (\nabla\cdot V_j)(\Phi_{t,s}^{\omega}(y))
    \circ\dd W_t^j .
\end{equation*}
Applying the Stratonovich form of the It\^o--Wentzell formula to the
time-dependent random field in the preceding identity yields
\begin{equation*}
    \dd\psi_t
    =
    -K_{V_0}\psi_t\,\dd t
    -
    \sum_{j=1}^{m}
    K_{V_j}\psi_t\circ\dd W_t^j .
\end{equation*}
This is \eqref{eq:half-density-sde}.  Using
$H_j=-\ii K_{V_j}$ from \eqref{eq:KvN-Hamiltonians} gives
\eqref{eq:strat-sse}.  Hence $U_{t,s}^{\omega}$ is its solution operator
for regular data and, by continuity and testing, for arbitrary $L^2$
initial data in the weak sense.

The propagator is strongly continuous in time.  For
$f\in C_c^\infty(\R^d)$, this follows from local convergence of the flow,
its inverse, and its Jacobian as $t\to s$.  Density of
$C_c^\infty(\R^d)$ in $L^2(\R^d)$ and unitarity extend the conclusion to
every $f\in L^2(\R^d)$.

For a trace-class operator $\Gamma_0$, unitary invariance gives
\begin{equation*}
    \left\|
       U_t^\omega\Gamma_0(U_t^\omega)^\dagger
    \right\|_1
    =
    \|\Gamma_0\|_1.
\end{equation*}
Under the standing measurability assumptions, the expectation in
\eqref{eq:random-unitary-channel} is therefore a well-defined trace-norm
Bochner integral.  Each unitary conjugation is completely positive and
trace preserving, and these properties are preserved by averaging.
Consequently, $\mathcal E_t$ is a quantum channel, and its adjoint fixes the
identity.

It remains to verify the semigroup property.  Write
\begin{equation*}
    \mathcal U_{t,s}^{\omega}(\Gamma)
    =
    U_{t,s}^{\omega}\Gamma(U_{t,s}^{\omega})^\dagger .
\end{equation*}
The cocycle gives
$\mathcal U_{t+s,0}^{\omega}
=\mathcal U_{t+s,s}^{\omega}\mathcal U_{s,0}^{\omega}$.
Because the coefficients are time independent and Brownian increments are
stationary, $\mathcal U_{t+s,s}^{\omega}$ has the same law as
$\mathcal U_{t,0}^{\omega}$.  Independence of the future increments from
$\mathcal F_s^W$ then gives
\begin{align*}
    \mathcal E_{t+s}(\Gamma)
    =
    \E\!\left[
       \E\!\left[
          \mathcal U_{t+s,s}^{\omega}
          \bigl(\mathcal U_{s,0}^{\omega}(\Gamma)\bigr)
          \,\middle|\,\mathcal F_s^W
       \right]
    \right] 
    =
    \mathcal E_t\!\left(\mathcal E_s(\Gamma)\right).
\end{align*}
This is the classical Markov property expressed at the channel level.
Strong continuity of $U_t^\omega$, finite-rank approximation, and dominated
convergence give trace-norm continuity of $\mathcal E_t$.  Hence
$(\mathcal E_t)_{t\geq0}$ is a strongly continuous random-unitary quantum
dynamical semigroup.  The weak generator identification is established by
the scalar It\^o calculation in Appendix~\ref{app:pathwise}.

\end{proof}

\subsection{Global flows for the polynomial examples}
\label{app:example-flow-completeness}

The globally Lipschitz sufficient condition stated in Section~\ref{sec:KL} does
not cover the polynomial drifts in the experiments.  Their two-sided flow
property can nevertheless be checked directly.  For the underdamped
double-well system \eqref{eq:numerical-double-well-sde}, fix a continuous
noise path \(\eta(t)=\sqrt{2\gamma/\beta}\,W_t\) and set
\(r=v-\eta(t)\).  The transformed equations are the ordinary,
time-dependent system
\[
    \dot q=r+\eta(t),\qquad
    \dot r=-U'(q)-\gamma(r+\eta(t)).
\]
For \(E(q,r)=1+r^2/2+U(q)\), \(U(q)=(q^2-1)^2/4\), cancellation of
the Hamiltonian terms gives
\[
    \dot E=\eta(t)U'(q)-\gamma r(r+\eta(t)),\qquad
    |\dot E|\leq C_M E
    \quad\text{if }\sup_{[0,T]}|\eta|\leq M.
\]
Here \(|U'|\leq C(1+U)\) and \(E\) is coercive.  Gr\"onwall's
inequality in both time directions prevents finite-time explosion for
every initial or terminal point.

For noisy Lorenz--63 \eqref{eq:numerical-lorenz-sde}, subtract
\(\eta(t)=\sqrt{2\kappa}\,W_t\) and write the transformed state as
\(z\).  The quadratic part of the drift is
\(Q(z)=(0,-z_1z_3,z_1z_2)\), so \(z\cdot Q(z)=0\).
Expanding \(b(z+\eta)\) leaves, besides \(Q(z)\), only terms of degree
at most one in \(z\), with coefficients bounded on a finite time
interval.  Therefore
\(\left|\frac{\mathrm d}{\mathrm dt}(1+|z|^2)\right|
\leq C_M(1+|z|^2)\), again in both time directions.  Local smooth
dependence on initial data and the variational equation then give global
smooth diffeomorphisms in both examples.  Setting \(\eta=0\) also proves
completeness of their drift fields; their constant noise fields are
complete translations.  This checks the flow hypotheses, without
asserting uniform-in-time Sobolev or Galerkin error bounds.

\section{Proof of the Kolmogorov--Lindblad representation theorem}
\label{app:pathwise}

This section proves Theorem~\ref{prop:KL}.  The proof has two simple ingredients.
Stratonovich calculus obeys the ordinary chain rule, so the modulus square
of the half-density satisfies the same stochastic continuity equation as
the pathwise probability density.  After conversion to It\^o form, the
quadratic variation of the Brownian terms becomes exactly the Lindblad
dissipator.

\begin{proof}
We first prove \eqref{eq:stochastic-continuity} in weak form, by pairing
with a test function.
Let $f$ be a smooth compactly supported test function.  For a Brownian
realisation for which the flow is defined, Stratonovich calculus gives
\begin{equation*}
    \dd f(\Phi_t^\omega(y))
    =
    V_0\cdot\nabla f(\Phi_t^\omega(y))\,\dd t
    +
    \sum_{j=1}^{m}
    V_j\cdot\nabla f(\Phi_t^\omega(y))
    \circ\dd W_t^j.
\end{equation*}
Multiplying by $p_0(y)$, integrating over the initial point and moving the time evolution to $p_t^\omega$ gives
\begin{equation*}
    \dd\int_{\R^d} f(x)p_t^\omega(x)\,\dd x
    =
    \int_{\R^d}V_0\cdot\nabla f\,p_t^\omega\,\dd x\,\dd t
    +
    \sum_{j=1}^{m}
    \int_{\R^d}V_j\cdot\nabla f\,p_t^\omega\,\dd x
    \circ\dd W_t^j.
\end{equation*}

The half-density in \eqref{eq:half-density-pullback} is defined from the
same stochastic flow.  Comparing its square-root Jacobian with the
pushforward formula for $p_t^\omega$ gives
\begin{equation*}
    |\psi_t(x)|^2=p_t^\omega(x)
\end{equation*}
directly, without appealing to uniqueness of the stochastic continuity
equation.  The differential version of this identity follows from the
ordinary Stratonovich product rule,
\begin{equation*}
    \dd|\psi_t|^2
    =
    2\operatorname{Re}(\overline{\psi_t}\,\dd\psi_t).
\end{equation*}
For any real vector field $V$,
\begin{equation*}
    2\operatorname{Re}(\overline\psi K_V\psi)
    =
    V\cdot\nabla|\psi|^2
    +(\nabla\cdot V)|\psi|^2
    =
    \nabla\cdot(V|\psi|^2).
\end{equation*}
Consequently,
\begin{equation*}
    \dd|\psi_t|^2
    =
    -\nabla\cdot(V_0|\psi_t|^2)\,\dd t
    -
    \sum_{j=1}^{m}
    \nabla\cdot(V_j|\psi_t|^2)\circ\dd W_t^j.
\end{equation*}
Thus $|\psi_t|^2$ also satisfies \eqref{eq:stochastic-continuity}.  The
normalisation
$\|\psi_t\|_2=1$ follows either by integrating this identity or from the
unitarity proved in Appendix~\ref{app:random-unitary}.

To derive the density-matrix equation, rewrite \eqref{eq:strat-sse} in It\^o
form as
\begin{equation*}
    \dd\psi_t
    =
    \left(
       -\ii H_0-
       \frac12\sum_{j=1}^{m}H_j^2
    \right)\psi_t\,\dd t
    -
    \ii\sum_{j=1}^{m}H_j\psi_t\,\dd W_t^j.
\end{equation*}
Applying the Hilbert-space It\^o product rule to
$P_t=\ket{\psi_t}\bra{\psi_t}$ and using
$\dd W_t^j\dd W_t^k=\delta_{jk}\dd t$ gives
\begin{equation*}
    \dd P_t
    =
    \left[
       -\ii[H_0,P_t]
       +
       \sum_{j=1}^{m}
       \left(
          H_jP_tH_j
          -\frac12\{H_j^2,P_t\}
       \right)
    \right]\dd t
    -
    \ii\sum_{j=1}^{m}[H_j,P_t]\dd W_t^j.
\end{equation*}
Taking expectations removes the martingale terms and proves
\eqref{eq:KL-standard}, hence the double-commutator form \eqref{eq:KL}.
This trace-class differential calculation requires regularity and
integrability of the displayed terms.  For arbitrary \(L^2\) initial
data one instead pairs the weak half-density equation with fixed smooth
test vectors and applies the scalar It\^o product rule.  For a finite-rank
test \(F\) built from \(\mathcal C\), the result is
\[
    \dd\Tr(FP_t)
    =\Tr(\mathcal L^\dagger(F)P_t)\dd t
      +\ii\sum_j\Tr([H_j,F]P_t)\dd W_t^j.
\]
All the test operators on the right extend to bounded finite-rank
operators.  Since \(\|\psi_t\|_2=\|\psi_0\|_2\), the stochastic
integrands are bounded, and hence are genuine square-integrable
martingales on finite time intervals.  Expectation proves
\eqref{eq:KL-weak-form} for rank-one states.  Polarization, finite-rank
approximation, and trace-norm continuity extend it to every trace-class
initial operator.  This argument avoids imposing the domains of the
unbounded Hamiltonians on arbitrary initial states.

Positivity and unit trace follow either from
$\Gamma(t)=\E_\omega[P_t]$ or from Theorem~\ref{thm:random-unitary}.  Moreover,
for every bounded measurable function $\phi$,
\begin{align*}
    \Tr(M_\phi\Gamma(t))
    =
    \E_\omega
    \int_{\R^d}\phi(x)|\psi_t(x)|^2\,\dd x
    =
    \int_{\R^d}\phi(x)p(t,x)\,\dd x
    =
    \E[\phi(X_t)].
\end{align*}
By the defining identity \eqref{eq:position-diagonal-map}, this also proves
$\mathsf D\Gamma(t)=p(t,\cdot)$.  All four statements of
Theorem~\ref{prop:KL} follow.
\end{proof}

\section{Proof of the intertwining theorem and regression corollary}
\label{app:intertwining}
\label{app:time-correlation}

This section proves Theorem~\ref{prop:intertwine} and then derives
Corollary~\ref{cor:KLM-regression}.  For the theorem, the key observation is that a
first-order KLM generator differentiates the two variables of an operator
kernel in parallel.  On the position diagonal, the two derivatives combine
by the chain rule and reproduce the classical continuity operator.  The
noise term follows by applying this first-order identity twice.  The
semigroup statement then follows directly from the stochastic flow.

\begin{proof}
We first prove the semigroup identities, which do not require a pointwise
kernel.  For every Brownian realisation, the half-density pullback and a
change of variables give
\begin{equation*}
    (U_t^\omega)^\dagger M_fU_t^\omega
    =
    M_{f\circ\Phi_t^\omega}.
\end{equation*}
Averaging and using the definition of the backward Markov semigroup yields
\begin{equation*}
    \mathcal E_t^\dagger(M_f)
    =
    M_{\mathsf P_tf}.
\end{equation*}

Recall that the backward Markov semigroup is defined by
\begin{equation*}
    (\mathsf P_t f)(x)
    =
    \E_\omega\!\left[f(\Phi_t^\omega(x))\right]
    =
    \E\!\left[f(X_t)\mid X_0=x\right],
\end{equation*}
and $\mathsf P_t^*$ denotes its adjoint acting on densities.  For every
trace-class $\Gamma$ and bounded measurable $f$,
\begin{align*}
    \int_{\R^d}f\,\mathsf D\mathcal E_t(\Gamma)\,\dd x
    &=
    \Tr\!\left(M_f\mathcal E_t(\Gamma)\right)
    =
    \Tr\!\left(\mathcal E_t^\dagger(M_f)\Gamma\right) \\
    &=
    \Tr\!\left(M_{\mathsf P_t f}\Gamma\right)
    =
    \int_{\R^d}(\mathsf P_t f)\,\mathsf D\Gamma\,\dd x
    =
    \int_{\R^d}f\,\mathsf P_t^*(\mathsf D\Gamma)\,\dd x .
\end{align*}
Here the third equality uses the Heisenberg intertwining
$\mathcal E_t^\dagger(M_f)=M_{\mathsf P_t f}$ proved above, while the final
equality is the defining duality between $\mathsf P_t$ and $\mathsf P_t^*$.
Since this holds for every bounded $f$, it proves both identities in
\eqref{eq:semigroup-intertwining}.

The generator identities hold by differentiation on generator domains.
Their differential form can also be verified directly on smooth tests,
without claiming strong continuity of the backward semigroup on
$L^\infty$.  For a real vector field $V$,
set
\begin{equation*}
    \Cop_V(\Gamma)=-\ii[H_V,\Gamma],
    \qquad H_V=-\ii K_V.
\end{equation*}
On a finite-rank kernel $\gamma_\Gamma(x,y)$ built from the common test
domain,
\begin{equation}
    (\Cop_V\Gamma)(x,y)
    =
    -\Bigl[
       V(x)\cdot\nabla_x+V(y)\cdot\nabla_y
       +\frac12(\nabla\cdot V)(x)
       +\frac12(\nabla\cdot V)(y)
    \Bigr]\gamma_\Gamma(x,y).
    \label{eq:kernel-transport-superoperator}
\end{equation}
If $q(x)=\gamma_\Gamma(x,x)$, then
\begin{equation*}
    \nabla q(x)
    =
    \left.
       (\nabla_x+\nabla_y)\gamma_\Gamma(x,y)
    \right|_{y=x}.
\end{equation*}
Restricting \eqref{eq:kernel-transport-superoperator} to the diagonal gives
\begin{equation}
\begin{split}
    \mathsf D(\Cop_V\Gamma)
    &=
    -V\cdot\nabla q-(\nabla\cdot V)q
    =
    -\nabla\cdot(Vq)
    =
    \mathcal T_V(\mathsf D\Gamma).
\end{split}
    \label{eq:first-order-diagonal-intertwining}
\end{equation}
The test class is invariant under the operators used here.  We may therefore
apply this identity twice to each noise direction.  Since
\begin{equation*}
    \Lop
    =
    \Cop_{V_0}
    +
    \frac12\sum_{j=1}^{m}\Cop_{V_j}^{2},
\end{equation*}
we obtain
\begin{equation*}
    \mathsf D(\Lop\Gamma)
    =
    \Lop_{\mathrm{FP}}(\mathsf D\Gamma).
\end{equation*}
For the backward identity, the product rule gives
\([K_V,M_f]=M_{V\cdot\nabla f}\).  Thus
\(\ii[H_0,M_f]=M_{V_0\cdot\nabla f}\) and
\(-[H_j,[H_j,M_f]]=M_{(V_j\cdot\nabla)^2f}\), which proves
\(\mathcal L^\dagger(M_f)=M_{\mathcal A f}\) on smooth tests.
The semigroup identities established above extend the law and observable
relations beyond these differential domains.
\end{proof}

We now prove the time-correlation property, i.e, the KLM regression corollary.

The intuition is the classical Markov property.  The earlier observable
weights the law at time $s$, while the later observable is propagated
backward from $t$ to $s$.  The backward part of
Theorem~\ref{prop:intertwine} turns that classical propagation into the Heisenberg
evolution of a multiplication operator.

\begin{proof}
The Markov property gives
\begin{equation*}
\begin{split}
    \E\!\left[
       \phi_1(X_t)\phi_2(X_s)
    \right]
    =
    \int_{\R^d}
       (\mathsf P_{t-s}\phi_1)(x)
       \phi_2(x)p(s,x)\,\dd x.
\end{split}
\end{equation*}
Since $\mathsf D\Gamma(s)=p(s,\cdot)$, the right-hand side is
\begin{equation*}
    \Tr\!\left[
       M_{\mathsf P_{t-s}\phi_1}
       M_{\phi_2}\Gamma(s)
    \right].
\end{equation*}
The backward intertwining identity in Theorem~\ref{prop:intertwine} gives
\begin{equation*}
    M_{\mathsf P_{t-s}\phi_1}
    =
    \mathcal E_{t-s}^\dagger(M_{\phi_1}).
\end{equation*}
Therefore
\begin{align*}
    \E\!\left[
       \phi_1(X_t)\phi_2(X_s)
    \right]
    &=
    \Tr\!\left[
       \mathcal E_{t-s}^\dagger(M_{\phi_1})
       M_{\phi_2}\Gamma(s)
    \right]
    =
    \Tr\!\left[
       M_{\phi_1}
       \mathcal E_{t-s}
       \!\left(M_{\phi_2}\Gamma(s)\right)
    \right],
\end{align*}
which is \eqref{eq:KLM-regression}.  The expression depends on
$\Gamma(s)$ only through its position density because the product of the
two multiplication operators is again a multiplication operator.  This
also proves the stated independence from phase and off-diagonal coherence.
\end{proof}

\section{A standard Galerkin error argument}
\label{app:galerkin-error}

The weak-bias assumption in
Theorem~\ref{thm:KLM-observable-complexity} depends on the chosen basis, the
coefficients, and the regularity of the particular SDE.  We therefore do not
attempt a general approximation theorem here.  Instead, we outline the
standard semidiscrete Galerkin workflow.  One projects the continuum
equation, identifies the resulting consistency residuals, applies an energy
estimate and Gr\"onwall's inequality, and then inserts the approximation
properties of the trial space.  Similar arguments are standard for
stochastic Schr\"odinger and transport equations
\cite{AntonCohenLarssonWang2016,ChenHongJi2017,
FjordholmKarlsenPang2025}.  The result is stated precisely at the end of the section.

\paragraph{Projected equation and consistency residuals.}

Set
\begin{equation*}
    A_j=-\ii H_j,
    \qquad j=0,1,\ldots,m,
    \qquad
    B=A_0+\frac12\sum_{j=1}^{m}A_j^2.
\end{equation*}
Under the domain and integrability assumptions required for the following
Bochner and stochastic integrals, the It\^o equation
\eqref{eq:ito-sse} has the integral form
\begin{equation}
    \psi(t)
    =
    \psi_0
    +\int_0^tB\psi(s)\,\dd s
    +\sum_{j=1}^{m}\int_0^tA_j\psi(s)\,\dd W_s^j.
    \label{eq:sse-integral-appendix}
\end{equation}
Let $\Pi_h$ be the $L^2$-orthogonal projection onto $Y_h$ and define
\begin{equation*}
    A_{j,h}=\Pi_hA_j\Pi_h,
    \qquad
    B_h=A_{0,h}+\frac12\sum_{j=1}^{m}A_{j,h}^2.
\end{equation*}
Projecting the Stratonovich equation first and then converting it to It\^o
form gives
\begin{equation}
    \dd\psi_h
    =
    B_h\psi_h\,\dd t
    +
    \sum_{j=1}^{m}A_{j,h}\psi_h\,\dd W_t^j.
    \label{eq:projected-ito-appendix}
\end{equation}
We couple the continuum and projected equations using the same Brownian
motions.

Decompose the total error as
\begin{equation*}
    \psi-\psi_h=\eta_h+e_h,
    \qquad
    \eta_h=(I-\Pi_h)\psi,
    \qquad
    e_h=\Pi_h\psi-\psi_h.
\end{equation*}
Here $\eta_h$ is the approximation tail, while $e_h$ is the error generated
by evolution inside the projected space.  Inserting $\Pi_h\psi$ into
\eqref{eq:projected-ito-appendix} gives the noise and It\^o-drift residuals
\begin{align}
    r_{j,h}
    &=
    \Pi_hA_j\psi-A_{j,h}\Pi_h\psi
    =
    \Pi_hA_j(I-\Pi_h)\psi,
    \label{eq:noise-residual-appendix}
    \\
    r_{B,h}
    &=
    \Pi_hB\psi-B_h\Pi_h\psi.
    \label{eq:drift-residual-appendix}
\end{align}
Subtracting the projected equation from the projection of
\eqref{eq:sse-integral-appendix} yields
\begin{equation*}
    \dd e_h
    =
    (B_he_h+r_{B,h})\,\dd t
    +
    \sum_{j=1}^{m}
    (A_{j,h}e_h+r_{j,h})\,\dd W_t^j.
\end{equation*}

For the energy estimate it is useful to combine the drift residual with
the noise cross terms:
\begin{equation}
    R_h=r_{B,h}-\sum_{j=1}^m A_{j,h}r_{j,h}.
    \label{eq:combined-residual}
\end{equation}
We collect the projection and consistency errors in
\begin{equation}
    \mathcal E_h^2
    =
    \sup_{0\leq t\leq T}
    \E\|(I-\Pi_h)\psi(t)\|_2^2
    +
    \E\|\Pi_h\psi_0-\psi_h(0)\|_2^2
    +
    \E\int_0^T
    \left[
       \|R_h\|_2^2
       +\sum_{j=1}^{m}\|r_{j,h}\|_2^2
    \right]\dd t.
    \label{eq:galerkin-defect}
\end{equation}
This residual quantity is the useful basis-independent object.  A concrete
approximation space enters only when the terms in
\eqref{eq:galerkin-defect} are estimated.

\paragraph{Energy estimate.}

Since the structure-preserving projection gives
$A_{j,h}^\dagger=-A_{j,h}$, It\^o's formula yields
$$
    \dd\|e_h\|_2^2
    =
    \left[
       2\operatorname{Re}\langle e_h,r_{B,h}\rangle
       +
       \sum_{j=1}^{m}
       \left(
          \|r_{j,h}\|_2^2
          -
          2\operatorname{Re}
          \langle e_h,A_{j,h}r_{j,h}\rangle
       \right)
    \right]\dd t
    +
    2\sum_{j=1}^{m}
    \operatorname{Re}\langle e_h,r_{j,h}\rangle\,\dd W_t^j.
$$ 
All homogeneous energy terms cancel.  In particular,
\begin{equation*}
    2\operatorname{Re}
    \left\langle
       e_h,\frac12A_{j,h}^2e_h
    \right\rangle
    +
    \|A_{j,h}e_h\|_2^2
    =
    0.
\end{equation*}
The term $A_{j,h}r_{j,h}$ is not an additional residual.  It appears when
skew-Hermiticity moves $A_{j,h}$ from the unknown error onto the noise
residual.

After localisation if necessary, taking expectations removes the
martingale term. Writing $F_h(t)=\E\|e_h(t)\|_2^2$, Young's inequality gives
\[
 F_h'(t)\leq F_h(t)+\E\|R_h(t)\|_2^2
                    +\sum_{j=1}^m\E\|r_{j,h}(t)\|_2^2.
\]
Orthogonality gives $\|\psi-\psi_h\|_2^2=\|\eta_h\|_2^2+\|e_h\|_2^2$.
Gr\"onwall's lemma therefore yields
\begin{equation}
    \sup_{0\leq t\leq T}
    \left(
       \E\|\psi(t)-\psi_h(t)\|_2^2
    \right)^{1/2}
    \leq
    C_T\mathcal E_h,
    \label{eq:strong-galerkin-bound}
\end{equation}
where one may take $C_T=e^{T/2}$. This energy factor is independent
of $h$, $d$ and $m$; all dependence on the vector fields, the number of noise
channels and the approximation family is retained in $\mathcal E_h$.
Estimating $R_h$ by its separate summands can introduce additional
dimension- or noise-dependent constants, which must not be suppressed in
an end-to-end complexity estimate.
A stronger estimate with
$\E\sup_{t\leq T}\|\psi(t)-\psi_h(t)\|_2^2$ can be obtained under additional
integrability assumptions using the Burkholder--Davis--Gundy inequality, but
it is not needed for the terminal observables considered here.

\paragraph{A representative interpolation estimate.}

The residual estimate \eqref{eq:strong-galerkin-bound} does not assume a
particular convergence rate.  To illustrate how a conventional algebraic
rate arises, suppose that for an integer $s\geq3$,
\begin{equation}
    \|(I-\Pi_h)u\|_{H^q}
    \leq
    C_{\Pi,\ell,q}(d)h^{\ell-q}\|u\|_{H^\ell},
    \qquad
    0\leq q\leq2,
    \quad
    q\leq\ell\leq s.
    \label{eq:interpolation-appendix}
\end{equation}
Assume also the first-order mapping and inverse estimates
\begin{equation*}
    \|A_j u\|_{H^q}
    \leq
    C_A(d)\|u\|_{H^{q+1}},
    \qquad
    \|A_{j,h}v_h\|_2
    \leq
    C_A(d)h^{-1}\|v_h\|_2,
\end{equation*}
for the required values of $q$, together with the corresponding coefficient
and projection-stability bounds.  The source of the second-order
consistency requirement is visible from
\begin{equation*}
    r_{B,h}
    =
    \Pi_hA_0(I-\Pi_h)\psi
    +
    \frac12\sum_{j=1}^{m}
    \left[
       A_{j,h}r_{j,h}
       +
       \Pi_hA_j(I-\Pi_h)A_j\psi
    \right].
\end{equation*} 
If $ \sup_{0\leq t\leq T}
    \left(
       \E\|\psi(t)\|_{H^s}^2
    \right)^{1/2}
    \leq
    C_{s,T,V}(d)\|\psi_0\|_{H^s},$
and the projected initial state has the same approximation accuracy, then
substitution into \eqref{eq:galerkin-defect} gives
\begin{equation}
    \sup_{0\leq t\leq T}
    \left(
       \E\|\psi(t)-\psi_h(t)\|_2^2
    \right)^{1/2}
    \leq
    C_{s,T,V}(d)h^{s-2}\|\psi_0\|_{H^s}.
    \label{eq:integer-regularity-rate}
\end{equation}
For example, one may take
$\psi_h(0)=\Pi_h\psi_0/\|\Pi_h\psi_0\|_2$ whenever the projection is
nonzero.

The loss of two powers in \eqref{eq:integer-regularity-rate} is a
conservative consequence of the second-order It\^o correction and is not
claimed to be optimal.  The constant may depend on $d$, $T$, the
coefficients, and the approximation family.  For spectral or Hermite
spaces, \eqref{eq:interpolation-appendix} is replaced by the appropriate
spectral-tail or weighted-Sobolev estimate.  Establishing such estimates for
a particular nonlinear SDE is a separate approximation problem and is not
attempted here.

\paragraph{Transfer to KLM states and observables.}

Let
\begin{equation*}
    \Gamma(t)
    =
    \E\!\left[\ket{\psi(t)}\bra{\psi(t)}\right],
    \qquad
    \Gamma_h(t)
    =
    \E\!\left[\ket{\psi_h(t)}\bra{\psi_h(t)}\right],
\end{equation*}
where $\Gamma_h$ is embedded in $L^2(\R^d)$.  For normalised pure states,
convexity of the trace norm and the rank-one projector estimate give
\begin{equation}
    \|\Gamma(t)-\Gamma_h(t)\|_1
    \leq
    2\left(
       \E\|\psi(t)-\psi_h(t)\|_2^2
    \right)^{1/2}.
    \label{eq:trace-from-strong-appendix}
\end{equation}
For a normalised mixture propagated componentwise, the same argument is
applied to each component and then combined using convexity.

Since the position-density map $\mathsf D$ is contractive from trace class
to $L^1$,
\begin{equation}
    \|p(t)-p_h(t)\|_1
    \leq
    \|\Gamma(t)-\Gamma_h(t)\|_1.
    \label{eq:density-from-trace-appendix}
\end{equation}
Trace duality also gives, for every bounded operator $O$,
\begin{equation}
    \left|
       \Tr(O\Gamma(T))
       -
       \Tr(O\Gamma_h(T))
    \right|
    \leq
    2\|O\|
    \left(
       \E\|\psi(T)-\psi_h(T)\|_2^2
    \right)^{1/2}.
    \label{eq:bounded-observable-error}
\end{equation}
For a classical payoff, $O=M_\phi$ and
$\|O\|=\|\phi\|_\infty$.  If
$\Gamma_h=\widehat\Phi_hG_h\widehat\Phi_h^\dagger$ and
$O_{\phi,h}=\widehat\Phi_h^\dagger M_\phi\widehat\Phi_h$, then
\begin{equation*}
    \Tr(M_\phi\Gamma_h)
    =
    \Tr(O_{\phi,h}G_h).
\end{equation*}
Thus \eqref{eq:integer-regularity-rate} supplies the sufficient weak exponent
$r=s-2$ in \eqref{eq:KLM-weak-Galerkin-rate} for bounded payoffs.

Quadrature, far-field truncation, and state-preparation errors must be added
separately. 
The preceding calculation gives a residual estimate, rather than an
unconditional convergence rate for arbitrary trial spaces.

\begin{theorem}[Residual control of the Galerkin approximation]
\label{thm:galerkin-residual}
Let $\psi$ be a normalised solution of the continuum SSE and let $\psi_h$
be its normalised, structure-preserving Galerkin approximation, coupled by
the same Brownian motions. Assume that the operator domains and stochastic
integrability justify the projected error equation above and that
$\mathcal E_h$ in \eqref{eq:galerkin-defect} is finite. Then
\begin{align}
 \sup_{t\leq T}(\E\|\psi(t)-\psi_h(t)\|_2^2)^{1/2}
 &\leq e^{T/2}\mathcal E_h,\\
 \sup_{t\leq T}\|\Gamma(t)-\Gamma_h(t)\|_1
 &\leq 2e^{T/2}\mathcal E_h.
\end{align}
Consequently, for bounded $\phi$,
\[
 \left|\E\phi(X_T)-\Tr(O_{\phi,h}G_h(T))\right|
 \leq 2\|\phi\|_\infty e^{T/2}\mathcal E_h.
\]
If the consistency estimates give
$\mathcal E_h\leq C_{\rm res}(d,T)h^p$, this supplies the sufficient
weak-bias constant
$C_1(d,T,\phi)=2\|\phi\|_\infty e^{T/2}C_{\rm res}(d,T)$.
In particular, the interpolation and regularity hypotheses of
\eqref{eq:interpolation-appendix}--\eqref{eq:integer-regularity-rate}
give the conservative exponent $p=s-2$.
\end{theorem}

\begin{proof}
The energy estimate bounds the first quantity. The rank-one projector
inequality and Jensen's inequality give
\eqref{eq:trace-from-strong-appendix}; trace duality gives the observable
bound. Substitution of the assumed residual rate gives the final statement.
\end{proof}

For unbounded statistics, trace-norm convergence alone is insufficient.
For example, for a coordinate observable $x_k$ and a truncation radius $R$,
one can apply the theorem to $x_k\mathbf 1_{\{|x_k|\leq R\}}$ and add
\[
 R^{1-a}\left(\int |x_k|^a p(T,x)\,\dd x
              +\int |x_k|^a p_h(T,x)\,\dd x\right),\qquad a>1,
\]
provided these moments are controlled uniformly in the approximation.
This distinguishes the bounded-observable theorem from the empirical
coordinate means reported in the numerical experiments.

\section{Overdamped Langevin: operators and Gaussian twirling}
\label{subsec:overdamped-example}

To make the construction concrete, consider overdamped Langevin dynamics in
a smooth potential \(U\):
\begin{equation}
    \dd X_t
    =
    -\nabla U(X_t)\,\dd t
    +
    \sqrt{2\beta^{-1}}\,\dd W_t.
    \label{eq:overdamped-langevin}
\end{equation}
Such models are widely used in computational chemistry, materials science,
statistics, and machine learning.  We assume that the partition function 
\(Z=\int_{\R^d}e^{-\beta U(x)}\,\dd x\) is finite.  For the exact
random-unitary statement, we also assume that the drift obeys the standing
two-sided completeness condition.  A globally Lipschitz gradient is one
simple sufficient hypothesis.  

The density satisfies
\begin{equation}
    \partial_t p
    =
    \nabla\cdot(p\nabla U)
    +
    \beta^{-1}\Delta p.
    \label{eq:overdamped-fpe}
\end{equation}
Its invariant density is the Gibbs distribution
\begin{equation*}
    \pi_\beta(x)=Z^{-1}e^{-\beta U(x)}.
\end{equation*}
A weighted form of this Fokker--Planck operator underlies recent quantum
algorithms for reaction-rate calculations
\cite{KharaziAlkadriMandadapuWhaley2026ReactionRates}.  Indeed, writing
\(p=\pi_\beta^{1/2}u\) gives \(\partial_tu=H_\beta u\), where
\begin{equation}
    H_\beta
    =
    \beta^{-1}\Delta
    -
    \frac{\beta}{4}|\nabla U|^2
    +
    \frac12\Delta U.
    \label{eq:Hbeta}
\end{equation}
This operator has the sum-of-squares factorisation
\begin{equation}
    H_\beta
    =
    -\sum_{j=1}^{d}A_j^\dagger A_j,
    \qquad
    A_j
    =
    \beta^{-1/2}\partial_{x_j}
    +
    \frac{\sqrt\beta}{2}\partial_{x_j}U.
    \label{eq:Hbeta-sos}
\end{equation}
This factorisation leads to the Gaussian-LCHS representation used in
Ref.~\cite{KharaziAlkadriMandadapuWhaley2026ReactionRates}.

KLM gives a different, trace-preserving representation of the same
probability law.  Since the noise is additive, the It\^o and Stratonovich
drifts coincide:
\begin{equation*}
    V_0=-\nabla U,
    \qquad
    V_j=\sqrt{2\beta^{-1}}\,e_j,
    \qquad j=1,\ldots,d.
\end{equation*}
The corresponding Hermitian operators, with
\(P=(P_1,\ldots,P_d)=-\ii\nabla\), are
\begin{equation}
\begin{aligned}
    H_0
    &=
    -\frac12\left(\nabla U\cdot P+P\cdot\nabla U\right)
    =
    -\ii\left(
       -\nabla U\cdot\nabla
       -
       \frac12\Delta U
    \right),
    \\
    H_j
    &=
    \sqrt{2\beta^{-1}}\,P_j,
    \qquad
    P_j=-\ii\partial_{x_j}.
\end{aligned}
    \label{eq:overdamped-KL-operators}
\end{equation}
Equation~\eqref{eq:KL} therefore becomes
\begin{equation}
    \dot\Gamma
    =
    -\ii[H_0,\Gamma]
    -
    \beta^{-1}\sum_{j=1}^{d}
       [P_j,[P_j,\Gamma]].
    \label{eq:overdamped-KL}
\end{equation}
Its position-space diagonal solves \eqref{eq:overdamped-fpe}, as guaranteed
by Theorem~\ref{prop:intertwine}.

The weighted Fokker--Planck formulation and KLM exhibit distinct Gaussian
structures.  The first uses the sum-of-squares factorisation
\eqref{eq:Hbeta-sos}.  In KLM, the Hermitian momentum jumps give an exact
Gaussian twirl.  If
\begin{equation*}
    \Dop(\Gamma)
    =
    -\beta^{-1}\sum_{j=1}^{d}[P_j,[P_j,\Gamma]],
\end{equation*}
then
\begin{equation}
    e^{t\Dop}(\Gamma)
    =
    \E_{\bm z\sim\mathcal N(0,\,2\beta^{-1}tI)}
    \left[
       e^{-\ii\bm z\cdot P}
       \Gamma
       e^{\ii\bm z\cdot P}
    \right].
    \label{eq:Gaussian-twirl}
\end{equation}
Thus the diffusion channel is an average of unitary translations.  It is
dephasing in the momentum representation and ordinary diffusion in the
position representation.  The typical translation time is
\(O(\sqrt{t/\beta})\), which is the origin of the square-root-in-time scale
available to isolated Gaussian-twirling channels
\cite{ShangGuoRebentrostAspuruGuzikLiZhao2025FastForwardQPE}.  In general \(H_0\) does not commute with the momenta.
The complete KLM evolution is then a Brownian path-ordered average rather
than a single Gaussian integral.

\section{Implementation details}
\label{supp:implementation}

This section records the coordinate transformations, assembly formulas,
and sufficient operator estimates used in the implementation discussion.
They preserve the algebraic Lindblad structure, while efficient coherent
access remains a separate computational requirement.

\subsection{Nonorthogonal bases and the mass matrix}
For a nonorthogonal finite-element or spectral basis, let
\(\Phi_h^{\rm raw}:\C^{N_h}\to L^2(\R^d)\) collect the basis and let
\(M_h=(\Phi_h^{\rm raw})^\dagger\Phi_h^{\rm raw}\) be the
positive-definite mass matrix, and write
\(\psi_h=\Phi_h^{\rm raw}c_h\).  The weak Galerkin equation is
\begin{equation}
 \begin{aligned}
    \ii M_h\dd c_h
    &=
    \mathbf H_0c_h\,\dd t
    +
    \sum_{j=1}^{m}
    \mathbf H_jc_h\circ\dd W_t^j,
    \\
    \mathbf H_j
    &=
    (\Phi_h^{\rm raw})^\dagger H_j\Phi_h^{\rm raw},
    \qquad j=0,\ldots,m.
 \end{aligned}
    \label{eq:mass-matrix-sse}
\end{equation}
The square root of the mass matrix converts this generalized
coordinate equation into a standard quantum-state equation.  With
\begin{equation*}
    \chi_h=M_h^{1/2}c_h,
    \qquad
    \widehat\Phi_h
    =
    \Phi_h^{\rm raw}M_h^{-1/2},
\end{equation*}
the synthesis map is an isometry and the projected operators are
  $  \widehat H_{j,h}
    =
    M_h^{-1/2}\mathbf H_jM_h^{-1/2}.$
They are Hermitian, and
\eqref{eq:finite-dimensional-sse}--\eqref{eq:finite-dimensional-lindblad}
hold with
\(\Phi_h,H_{j,h}\) replaced by
\(\widehat\Phi_h,\widehat H_{j,h}\).  In particular,
\begin{equation*}
    \|\chi_h\|_2^2
    =
    c_h^\dagger M_hc_h
    =
    \|\psi_h\|_{L^2}^2.
\end{equation*}
Thus \(\chi_h\), not the raw coefficient vector \(c_h\), is the normalised
quantum state.  The observable matrix becomes
\begin{equation*}
    \widehat O_{\phi,h}
    =
    M_h^{-1/2}
    (\Phi_h^{\rm raw})^\dagger
    M_\phi
    \Phi_h^{\rm raw}
    M_h^{-1/2}.
\end{equation*}
Below, \(O_{\phi,h}\) denotes the observable in orthonormal coefficient
coordinates.  Thus it means \(\widehat O_{\phi,h}\) when a raw
nonorthogonal basis is used.

The mass-matrix transformation is exact algebraically.  If $M_h$ admits an
efficient block encoding and suitable spectral bounds, QSVT can approximate
$M_h^{-1/2}$, with a cost depending on $\kappa(M_h)$ and the target precision
\cite{gilyen2019quantum}.  We therefore use an orthonormal basis or assume that the mass
matrix is uniformly well conditioned or appropriately preconditioned.

\subsection{Hermitian quadrature assembly}
Quadrature gives one concrete route from the weak differential operators to
finite matrices.  Let
\(\{x_q,w_q\}_{q=1}^{Q}\) be a quadrature rule and let
\(\{\varphi_a\}_{a=1}^{N_h}\) be a basis of \(Y_h\).  The basis is
represented at the quadrature nodes by
\begin{equation*}
    (\Phi_Q)_{qa}
    =
    \sqrt{w_q}\,\varphi_a(x_q),   \;
    M_Q
    =
    \Phi_Q^\dagger\Phi_Q,  \;
    \widehat\Phi_Q
    =
    \Phi_QM_Q^{-1/2}.
\end{equation*}
For an orthonormal basis and an exact quadrature, \(M_Q=I\).  We assume
that \(M_Q\) is positive definite.  If \(V_{j,k,Q}\) is the diagonal
matrix of the sampled coefficient \(V_{j,k}(x_q)\), and
\(P_{k,Q}\) is a quadrature-compatible Hermitian representation of
\(-\ii\partial_{x_k}\), then the Galerkin matrix is assembled as $
    H_{j,h} \approx H_{j,h}^{(Q)}:$
\begin{equation}
\begin{aligned}
    H_{j,h}^{(Q)}
    ={}&
    \frac12\sum_{k=1}^{d}
    \widehat\Phi_Q^\dagger
    \bigl(
       V_{j,k,Q}P_{k,Q}
       +
       P_{k,Q}V_{j,k,Q}
    \bigr)
    \widehat\Phi_Q,
\end{aligned}
    \label{eq:block-encoded-half-density}
\end{equation}
A summation-by-parts or compatible spectral differentiation matrix is one
way to preserve Hermiticity at the quadrature level.  A sufficiently
accurate or padded quadrature is needed to control aliasing near a spectral
cutoff.  In what follows, \(H_{j,h}\) denotes either the exact Galerkin
compression or its Hermitian quadrature approximation
\(H_{j,h}^{(Q)}\).  In the latter case, \(G_h\) denotes the corresponding
discrete channel and the quadrature bias is included in \(C_{\rm app}\).
The aliasing, mass-matrix, and far-field errors are likewise included in the
spatial error below.

\subsection{A first-order Hamiltonian dilation}
A simple dilation explains why finite-dimensional Lindblad dynamics is
natural on a quantum computer.  For one time step $\Delta t$, introduce an
$(m+1)$-level ancilla initialized in $\ket{0}$ and define
\begin{equation}
    \widetilde H_{h,\Delta t}
    =
    \begin{pmatrix}
       \sqrt{\Delta t}\,H_{0,h}
       & H_{1,h} & \cdots & H_{m,h}
       \\
       H_{1,h} & 0 & \cdots & 0
       \\
       \vdots & \vdots & \ddots & \vdots
       \\
       H_{m,h} & 0 & \cdots & 0
    \end{pmatrix}.
    \label{eq:KLM-star-dilation}
\end{equation}
Here the identical first row and column reflect the Hermiticity of the KLM
jump operators.  Evolve the joint system for time $\sqrt{\Delta t}$ with $U_{h,\Delta t}
    :=
    e^{-\ii\sqrt{\Delta t}\widetilde H_{h,\Delta t}}$ and then
discard and reset the ancilla:
\begin{equation}
\begin{split}
    \mathcal M_{h,\Delta t}(G)
    &=
    \Tr_{\rm a}\!\left[
       U_{h,\Delta t}
       \bigl(\ket{0}\bra{0}\otimes G\bigr)
       U_{h,\Delta t}^{\dagger}
    \right] .
\end{split}
\label{eq:KLM-first-order-dilation}
\end{equation}
A Taylor expansion gives
\[
    \mathcal M_{h,\Delta t}(G)
    =
    G+\Delta t\,\mathcal L_h(G)+O(\Delta t^2).
\]
Thus every step is exactly CPTP and requires no postselection, while repeated
steps give a first-order approximation to the Lindblad evolution.  Operationally, the trace-and-reset step is realized by mid-circuit
measurement and reset, a now-standard dynamic-circuit capability on several
quantum platforms.  Because the measurement outcome is discarded, neither
conditional feedback nor postselection is required. 
This repeated-interaction picture was described by
Cleve and Wang~\cite{CleveWang2017Lindblad}, while the systematic
Hamiltonian-dilation construction and its higher-order extensions were
developed by Ding, Li, and Lin~\cite{DingLiLin2024Hamiltonian}.

\subsection{Derivative bounds and access normalisation}
For the norm estimates, define the joint-jump map
$\mathsf H_{J,h}u=\sum_{j=1}^{m}\ket{j}\otimes H_{j,h}u$.
The normalisation can be related to the vector fields of the original SDE.
Define the full gradient restricted to the trial space by
\begin{equation*}
    \mathsf G_hu_h
    =
    \sum_{k=1}^{d}
    \ket{k}\otimes(-\ii\partial_{x_k})u_h,
    \qquad u_h\in Y_h.
\end{equation*}
We use \(h\) as an effective resolution parameter through the inverse bound
\begin{equation}
    \|\mathsf G_h\|
    \leq
    C_g\frac{\sqrt d}{h}.
    \label{eq:discrete-gradient-scale}
\end{equation}
For this scaling estimate, assume on the computational region that
\begin{equation*}
\begin{gathered}
    \|V_0\|_{\infty,2},
    \,
    \|\sigma\|_{\infty,\mathrm{op}},
    \,
    \|\nabla\cdot V_0\|_\infty,
    \,
    \left\|
       \bigl(\nabla\cdot\sigma_j\bigr)_{j=1}^{m}
    \right\|_{\infty,2}
    =
    O(1).
\end{gathered}
\end{equation*}
Here
\(\|f\|_{\infty,2}=\operatorname*{ess\,sup}_x|f(x)|_2\) and
\(\|\sigma\|_{\infty,\mathrm{op}}
=\operatorname*{ess\,sup}_x\|\sigma(x)\|_{2\to2}\).
Here \(V_0=b-\frac12\sum_j(D\sigma_j)\sigma_j\), so access to the
Stratonovich correction and its required coefficient derivatives is part of
the coefficient oracle. 

Combining the coefficient bounds with
\eqref{eq:discrete-gradient-scale} gives rough operator-norm estimates
\begin{equation}
    \|H_{0,h}\|
    =
    O\!\left(\frac{\sqrt d}{h}\right),
    \qquad
    \|\mathsf H_{J,h}\|
    =
    O\!\left(\frac{\sqrt d}{h}\right).
    \label{eq:KLM-block-normalisations}
\end{equation}
They imply the trace-norm estimate
\begin{equation}
    \|\mathcal L_h\|_{1\to1}
    \leq 2\|H_{0,h}\|+2\|\mathsf H_{J,h}\|^2
    =O(d/h^2),\qquad h\leq1.
    \label{eq:KLM-Lindblad-operator-scale}
\end{equation}
This is the expected parabolic operator scale, but an operator-norm bound
alone does not establish a matching block-encoding normalisation.
The access assumption in the main text is
$\alpha_{\mathcal L,h}\leq C_2(d)h^{-2}$ for $h\leq1$.
For example, $\alpha_{0,h}=O(\sqrt d/h)$ and
$\sum_j\alpha_{j,h}^2=O(d/h^2)$ give $C_2(d)=O(d)$.
The latter condition holds for standard separately encoded coordinate
derivatives with bounded diffusion coefficients, but is an additional
access condition in general.  Even if each separate encoding is
norm matched, the displayed coefficient assumptions alone give the
conservative bound $\sum_j\alpha_{j,h}^2=O(md/h^2)$, since
$\sum_j\|H_{j,h}\|^2\leq m\|\mathsf H_{J,h}\|^2$.
An efficiently encoded joint-jump map can potentially improve this bound,
but using that stronger access requires a compatible simulation
construction and is not assumed in the complexity theorem in the main text.

\section{Numerical methods and resolution diagnostics}
\label{supp:numerical-methods}
This section gives the detailed procedures underlying the numerical results
in the main text. The reported computations and values are those of the
supplied manuscript; no new simulations were performed for this revision.

\subsection{Double-well reference and Hermite approximation}

\label{subsubsec:numerical-double-well-discretisation}

To examine accuracy, we use an independent Fokker--Planck calculation on
the box
\[
    (q,v)\in[-3,3]\times[-5,5]
\]
with cell-centred uniform grids of \(128^2\), \(192^2\), and \(256^2\)
cells.  On an \(n^2\) grid,
\(\Delta q=6/n\), \(\Delta v=10/n\), and
\(h_{\rm FP}=\max\{\Delta q,\Delta v\}\).  A
second-order symmetric splitting separates velocity diffusion, friction,
potential force, and position transport
\cite{leimkuhler2015molecular,schlick2010molecular}.  If
\(\mathcal D_\tau\), \(\mathcal R_\tau\), \(\mathcal F_\tau\), and
\(\mathcal Q_\tau\) denote these four subflows, respectively, one step is
the symmetric composition
\begin{equation*}
 \begin{split}
    \mathcal S_{\Delta t}
    ={}&
    \mathcal D_{\Delta t/2}
    \mathcal R_{\Delta t/2}
    \mathcal F_{\Delta t/2}
    \mathcal Q_{\Delta t}
    \mathcal F_{\Delta t/2}
    \mathcal R_{\Delta t/2}
    \mathcal D_{\Delta t/2}.
 \end{split}
\end{equation*}
The transport subflows
\begin{equation*}
    \dot v=-\gamma v,
    \qquad
    \dot v=-(q^3-q),
    \qquad
    \dot q=v
\end{equation*}
are advanced by the exact maps
\(v\mapsto e^{-\gamma\tau}v\),
\(v\mapsto v-(q^3-q)\tau\), and
\(q\mapsto q+v\tau\), followed by a conservative
piecewise-linear remapping to the fixed grid.  Values outside the reference
box are set to zero during transport, and the removed cell mass is accumulated
as an escaped-mass diagnostic.  Velocity diffusion is advanced by the exact
semigroup of the periodic second-order discrete Laplacian, evaluated with a
Fourier transform.  

We use \(\Delta t_{\rm FP}=10^{-3}\).  Here the boundary-shell mass means
the total mass in cells touching the boundary of the box.  On the \(256^2\)
grid, its value at \(T=1\) is below \(1.1\times10^{-15}\), and the mass
removed by all transport steps is below \(1.9\times10^{-13}\).  Scalar
reference values are extrapolated linearly in \(h_{\rm FP}^2\) from the
three grids, where \(h_{\rm FP}\) is the larger coordinate spacing.  The
reported empirical uncertainty is the larger of the extrapolation correction
on the finest grid and the difference between the two finest grids. 

For the SSE and Lindblad calculations, we use centred and scaled Hermite functions
\begin{equation}
\phi_n^{(c,\ell)}(x)
=
\ell^{-1/2}
\frac{
\mathrm H_n\left((x-c)/\ell\right)
}{
\left(2^n n!\sqrt{\pi}\right)^{1/2}
}
\exp\left[-\frac{(x-c)^2}{2\ell^2}\right],
\label{eq:scaled-Hermite-functions}
\end{equation}
where \(\mathrm H_n\) is the \(n\)th physicists' Hermite polynomial.  The
functions are orthonormal in \(L^2(\R)\).  The parameter \(c\) translates the basis centre, while \(\ell\) sets its spatial
scale.  At fixed cutoff, a smaller \(\ell\) gives finer local resolution near
\(c\), while a larger \(\ell\) covers a broader region.  We take
\((\ell_q,\ell_v)=(0.35,0.50)\), centre both coordinates at zero, and use the tensor-product space
\begin{equation}
V_K
=
\operatorname{span}
\bigl\{
\phi_{n_q}^{(0,\ell_q)}(q)
\phi_{n_v}^{(0,\ell_v)}(v):
0\le n_q,n_v\le K
\bigr\},
\label{eq:numerical-double-well-space}
\end{equation}
whose dimension is
\begin{equation}
N=(K+1)^2 .
\label{eq:numerical-double-well-dimension}
\end{equation}
Let \(\Phi_K:\C^N\to V_K\) collect the basis functions and let
\(\Pi_K=\Phi_K\Phi_K^\dagger\).  In the notation of the preceding analysis,
\(Y_h=V_K\).  Since differentiation on the first \(K+1\) Hermite modes has
norm \(O(\sqrt{K+1}/\ell)\), the corresponding effective resolution
satisfies, up to constants depending only on the fixed dimension \(d=2\),
\begin{equation}
h_K^{-1}
\asymp
\sqrt{K+1}
\max\{\ell_q^{-1},\ell_v^{-1}\}.
\label{eq:Hermite-effective-resolution}
\end{equation}
Here \(h_K^{-1}\) is the effective derivative scale, while \(h_K\) is an
effective resolution length which might be different from a physical mesh spacing.

Polynomial operators are assembled on a padded Hermite space before being compressed to \(V_K\).  To explain why this is necessary, let \(Q\) denote multiplication by the position coordinate,
\[
(Qf)(q,v)=qf(q,v),
\]
and let \(\Pi_K\) be the orthogonal projection onto \(V_K\).  The Hermite recurrence relation gives
\[
Q\phi_n^{(c,\ell)}
=
c\phi_n^{(c,\ell)}
+
\frac{\ell}{\sqrt2}
\left(
\sqrt{n+1}\,\phi_{n+1}^{(c,\ell)}
+
\sqrt n\,\phi_{n-1}^{(c,\ell)}
\right),
\]
so each application of \(Q\) changes the Hermite level by at most one.  Because the double-well force contains a cubic term in \(q\), its exact Galerkin matrix requires
\[
\Pi_KQ^3\Pi_K .
\]
The code forms every operator product before the final compression.  In
particular, it evaluates \(Q^3\) in a space padded by at least three position
modes and applies \(\Pi_K\) only at the end.  By contrast,
\[
(\Pi_KQ\Pi_K)^3
\]
inserts a projection after every multiplication by \(q\).  It therefore discards paths that temporarily leave \(V_K\) and subsequently return, producing incorrect matrix elements near the spectral cutoff.

Writing
\begin{equation*}
    H_K=\Phi_K^\dagger H_0\Phi_K,
    \qquad
    L_K=\Phi_K^\dagger L\Phi_K,
\end{equation*}
we assemble the symmetrized weak forms directly, so both coefficient matrices
are Hermitian.  As in Section~\ref{subsec:galerkin}, we project the Stratonovich SSE
first and then convert the finite-dimensional equation to It\^o form.  The semidiscrete SSE is
\begin{equation}
    \dd \chi_t
    =
    \left(-\ii H_K-\frac12L_K^2\right)\chi_t\dd t
    -\ii L_K\chi_t\dd W_t ,
    \label{eq:numerical-double-well-sse}
\end{equation}
For a pure initial encoding, its covariance
\(G_K=\E[\chi_t\chi_t^\dagger]\) obeys
\begin{equation}
    \dot G_K
    =
    -\ii[H_K,G_K]
    +
    L_KG_KL_K-\frac12\{L_K^2,G_K\}.
    \label{eq:numerical-double-well-discrete-lindblad}
\end{equation}

If \(\boldsymbol\Phi_K(q,v)\) is the column vector of real tensor-product
basis functions spanning \(V_K\), the reconstructed diagonal density is
\begin{equation*}
    p_K(t,q,v)
    =
    \boldsymbol\Phi_K(q,v)^\dagger
    G_K(t)
    \boldsymbol\Phi_K(q,v).
\end{equation*}
The reported density error is the cell-quadrature approximation of
\(\|p_{\rm FP}-p_K\|_{L^1(B)}\) over the stated reference box \(B\), with
the Hermite density evaluated at the Fokker--Planck cell centres and without
clipping or renormalizing either density on that box.

We integrate \eqref{eq:numerical-double-well-discrete-lindblad} with
classical fourth-order Runge--Kutta (RK4), using
\(\Delta t_{\rm L}=10^{-3}\) and the cutoffs
\[
    K=8,10,12,14,16,18,20 .
\]
RK4 is used only as a classical verification method.  It is not a CPTP time
integrator, so positivity does not follow from the Lindblad form of the
semidiscrete equation.   In the reported
runs, trace and Hermiticity are preserved to roundoff and no eigenvalue more negative than the roundoff tolerance is observed.  Halving the RK4 step
at \(K=18\) leaves the two reported observables unchanged at the displayed
precision.

We choose the two observables as the mean position and the right-well
probability.  With \(Q_K=\Phi_K^\dagger Q\Phi_K\),
\begin{equation}
    \mu_{q,K}(T)
    :=
    \Tr(Q_KG_K(T)),
    \qquad
    P_{R,K}(T)
    :=
    \Tr(O_{R,K}G_K(T)),
    \label{eq:numerical-double-well-observables}
\end{equation}
where
\[
    [O_{R,K}]_{mn}
    =
    \int_0^\infty\phi_m^{(0,\ell_q)}(q)
                      \phi_n^{(0,\ell_q)}(q)\dd q
\]
in the position factor and \(O_{R,K}\) acts as the identity in velocity.
These approximate $\E[q_T]$ and $\mathbb P(q_T>0)$, respectively.
The event probability is evaluated directly through its projected observable.

\subsection{Lorenz--63 reference and Hermite approximation}
The Fokker--Planck reference uses the box
\[
    [-22,6]\times[-26,8]\times[8,46]
\]
and grids of \(64^3\), \(80^3\), and \(96^3\) cells.  Diffusion is again
advanced by the exact semigroup of the periodic discrete Laplacian.  With the
other coordinates frozen, each of the three Lorenz coordinate subflows is
affine in its active coordinate and therefore has an exact characteristic
map.  We compose the conservative remaps in the symmetric sequence
\(X/2\)--\(Y/2\)--\(Z\)--\(Y/2\)--\(X/2\), surrounded by diffusion half
steps.  The full advection--diffusion step uses
\(\Delta t_{\rm FP}=10^{-3}\).  Halving this step on the \(96^3\) grid
changes \(\E[X_T]\) by \(4.5\times10^{-6}\).  The final boundary-shell mass
and accumulated escaped mass are at roundoff level.  The reference mean is
extrapolated linearly in \(h_{\rm FP}^2\) from the three grids.  Its empirical
spatial uncertainty is defined by the same refinement rule as in the
double-well calculations, while the time-step change is reported separately.

As in the Langevin example, the additive noise produces Hermitian momentum
jumps.  Their dissipators can be regarded as dephasing channels. 
For the SSE/Lindblad calculation, we use the total-degree Hermite space
\begin{equation}
    V_K
    =
    \operatorname{span}
    \left\{
       \prod_{j=1}^3
       \phi_{\alpha_j}^{(c_j,\ell_j)}(x_j):
       |\alpha|\le K
    \right\},
    \label{eq:numerical-lorenz-space}
\end{equation}
with
\begin{equation}
    c=(-8,-8,27),
    \qquad
    \ell=(\sqrt2\,1.5,\sqrt2\,1.5,\sqrt2\,1.5).
    \label{eq:numerical-lorenz-basis-parameters}
\end{equation}
This choice makes \(\sqrt{p_0}\) the tensor Hermite ground state, so the
initial projection mass is one for every cutoff.  Let \(\Phi_K\) denote the
synthesis isometry for this total-degree space, and define projected
Hamiltonians and observables by direct coefficient compression with
\(\Phi_K\).  The space dimension is
\begin{equation}
    N=\binom{K+3}{3},
    \label{eq:numerical-lorenz-dimension}
\end{equation}
and we take \(K=2,4,\ldots,14\).  The quadratic drift operators are formed on
a degree-\((K+3)\) padded space before compression.  We integrate the
finite-dimensional Lindblad equation by RK4 with
\(\Delta t_{\rm L}=3.5\times10^{-4}\) and halving this step at \(K=12\)
changes \(\E[X_T]\) by less than \(2\times10^{-10}\).  As in the Langevin
calculations, we monitor trace, Hermiticity, and the smallest eigenvalue
because RK4 itself is not CPTP.
The reported mean is evaluated as
\begin{equation*}
    X_K=\Phi_K^\dagger M_x\Phi_K,
    \qquad
    \E[X_T]_K=\Tr(X_KG_K(T)).
\end{equation*}

\bibliography{sde}

@article{giles2008multilevel,
  title={Multilevel monte carlo path simulation},
  author={Giles, Michael B},
  journal={Operations Research},
  volume={56},
  number={3},
  pages={607--617},
  year={2008},
  publisher={INFORMS},
  doi = {10.1287/opre.1070.0496},
  note = {\url{https://doi.org/10.1287/opre.1070.0496}}
}

@book{pope2000turbulent,
  title={Turbulent flows},
  author={Pope, Stephen B},
  year={2000},
  publisher={Cambridge University Press},
  doi       = {10.1017/CBO9780511840531},
  note = {\url{https://doi.org/10.1017/CBO9780511840531}}
}

@book{ermoliev1988numerical,
  editor    = {Yuri Ermoliev and Roger J.-B. Wets},
  title     = {Numerical Techniques for Stochastic Optimization},
  series    = {Springer Series in Computational Mathematics},
  volume    = {10},
  year      = {1988},
  publisher = {Springer-Verlag},
  address   = {Berlin},
  isbn      = {978-0-387-18677-1},
  url       = {https://link.springer.com/book/9783642648137}
}

@article{AnLindenLiuMontanaroShaoWang2021,
  title = {Quantum-accelerated multilevel Monte Carlo methods for stochastic differential equations in mathematical finance},
  author = {An, Dong and Linden, Noah and Liu, Jin-Peng and Montanaro, Ashley and Shao, Changpeng and Wang, Jiasu},
  journal = {Quantum},
  volume = {5},
  pages = {481},
  year = {2021},
  doi = {10.22331/q-2021-06-24-481},
  note = {\url{https://doi.org/10.22331/q-2021-06-24-481}}
}

@book{Pavliotis2014Stochastic,
  title = {Stochastic Processes and Applications: Diffusion Processes, the Fokker--Planck and Langevin Equations},
  author = {Pavliotis, Grigorios A.},
  publisher = {Springer},
  year = {2014},
  doi = {10.1007/978-1-4939-1323-7},
  note = {\url{https://doi.org/10.1007/978-1-4939-1323-7}}
}

@book{Kunita1990StochasticFlows,
  title = {Stochastic Flows and Stochastic Differential Equations},
  author = {Kunita, Hiroshi},
  publisher = {Cambridge University Press},
  year = {1990},
  doi = {10.1017/CBO9780511663405},
  note = {\url{https://doi.org/10.1017/CBO9780511663405}}
}

@book{van1992stochastic,
  author    = {van Kampen, N. G.},
  title     = {Stochastic Processes in Physics and Chemistry},
  volume    = {1},
  year      = {1992},
  publisher = {Elsevier},
  isbn      = {978-0-444-89349-9},
  url       = {https://books.google.com/books?id=3e7XbMoJzmoC}
}

@book{kloeden1992numerical,
  author    = {Kloeden, Peter E. and Platen, Eckhard},
  title     = {Numerical Solution of Stochastic Differential Equations},
  series    = {Applications of Mathematics},
  volume    = {23},
  year      = {1992},
  publisher = {Springer-Verlag},
  address   = {Berlin},
  doi       = {10.1007/978-3-662-12616-5},
  note = {\url{https://doi.org/10.1007/978-3-662-12616-5}}
}

@article{black1973pricing,
  title={The pricing of options and corporate liabilities},
  author={Black, Fischer and Scholes, Myron},
  journal={Journal of Political Economy},
  volume={81},
  number={3},
  pages={637--654},
  year={1973},
  publisher={The University of Chicago Press},
  doi = {10.1086/260062},
  note = {\url{https://doi.org/10.1086/260062}}
}

@book{sarkka2019applied,
  author    = {S{\"a}rkk{\"a}, Simo and Solin, Arno},
  title     = {Applied Stochastic Differential Equations},
  series    = {Institute of Mathematical Statistics Textbooks},
  volume    = {10},
  year      = {2019},
  publisher = {Cambridge University Press},
  doi       = {10.1017/9781108186735},
  note = {\url{https://doi.org/10.1017/9781108186735}}
}

@article{Applebaum1992Operator,
  author = {Applebaum, David},
  title = {An Operator Theoretic Approach to Stochastic Flows on Manifolds},
  journal = {S{\'e}minaire de Probabilit{\'e}s},
  volume = {26},
  pages = {514--532},
  year = {1992},
  url = {https://www.numdam.org/item/SPS_1992__26__514_0/}
}

@article{AntonCohenLarssonWang2016,
  author = {Anton, Rikard and Cohen, David and Larsson, Stig and Wang, Xiaojie},
  title = {Full Discretization of Semilinear Stochastic Wave Equations Driven by Multiplicative Noise},
  journal = {SIAM Journal on Numerical Analysis},
  volume = {54},
  number = {2},
  pages = {1093--1119},
  year = {2016},
  doi = {10.1137/15M101049X},
  note = {\url{https://doi.org/10.1137/15M101049X}}
}

@article{ChenHongJi2017,
  author = {Chen, Chuchu and Hong, Jialin and Ji, Lihai},
  title = {Mean-Square Convergence of a Symplectic Local Discontinuous {Galerkin} Method Applied to Stochastic Linear {Schr\"odinger} Equation},
  journal = {IMA Journal of Numerical Analysis},
  volume = {37},
  number = {2},
  pages = {1041--1065},
  year = {2017},
  doi = {10.1093/imanum/drw023},
  note = {\url{https://doi.org/10.1093/imanum/drw023}}
}

@article{FjordholmKarlsenPang2025,
  author = {Fjordholm, Ulrik S. and Karlsen, Kenneth H. and Pang, Peter H. C.},
  title = {Convergent Finite Difference Schemes for Stochastic Transport Equations},
  journal = {SIAM Journal on Numerical Analysis},
  volume = {63},
  number = {1},
  pages = {149--192},
  year = {2025},
  doi = {10.1137/23M159946X},
  note = {\url{https://doi.org/10.1137/23M159946X}}
}

@misc{watanabe2026end,
  title         = {End-to-End Molecular Dynamics with a Langevin Thermostat on Quantum Circuits},
  author        = {Watanabe, Masari and Nishi, Hirofumi and Kosugi, Taichi and Hidaka, Shigekazu and Sakurai, Ryo and Matsushita, Yu-ichiro},
  year          = {2026},
  eprint        = {2605.30143},
  archivePrefix = {arXiv},
  primaryClass  = {quant-ph},
  note = {\url{https://arxiv.org/abs/2605.30143}}
}

@article{kramers1940brownian,
  title={Brownian motion in a field of force and the diffusion model of chemical reactions},
  author={Kramers, Hendrik Anthony},
  journal={Physica},
  volume={7},
  number={4},
  pages={284--304},
  year={1940},
  publisher={Elsevier},
  doi ={10.1016/S0031-8914(40)90098-2},
  note = {\url{https://doi.org/10.1016/S0031-8914(40)90098-2}}
}

@article{hanggi1990reaction,
  title={Reaction-rate theory: fifty years after Kramers},
  author={H{\"a}nggi, Peter and Talkner, Peter and Borkovec, Michal},
  journal={Reviews of Modern Physics},
  volume={62},
  number={2},
  pages={251--341},
  year={1990},
  publisher={APS},
  doi = {10.1103/RevModPhys.62.251},
  note = {\url{https://doi.org/10.1103/RevModPhys.62.251}}
}

@book{risken1996fokker,
  author = {Risken, Hannes},
  title = {The {Fokker--Planck} Equation: Methods of Solution and Applications},
  series = {Springer Series in Synergetics},
  volume = {18},
  edition = {2},
  year = {1996},
  publisher = {Springer},
  address = {Berlin, Heidelberg},
  doi = {10.1007/978-3-642-61544-3},
  note = {\url{https://doi.org/10.1007/978-3-642-61544-3}}
}

@article{lorenz1963mechanics,
  author  = {Lorenz, Edward N.},
  title   = {Deterministic Nonperiodic Flow},
  journal = {Journal of the Atmospheric Sciences},
  volume  = {20},
  number  = {2},
  pages   = {130--141},
  year    = {1963},
  doi     = {10.1175/1520-0469(1963)020<0130:DNF>2.0.CO;2},
  note = {\url{https://doi.org/10.1175/1520-0469(1963)020<0130:DNF>2.0.CO;2}}
}

@article{gorini1976completely,
  title={Completely positive dynamical semigroups of N-level systems},
  author={Gorini, Vittorio and Kossakowski, Andrzej and Sudarshan, Ennackal Chandy George},
  journal={Journal of Mathematical Physics},
  volume={17},
  number={5},
  pages={821--825},
  year={1976},
  publisher={American Institute of Physics},
  doi = {10.1063/1.522979},
  note = {\url{https://doi.org/10.1063/1.522979}}
}

@article{lindblad1976generators,
  title={On the generators of quantum dynamical semigroups},
  author={Lindblad, G{\"o}ran},
  journal={Communications in Mathematical Physics},
  volume={48},
  number={2},
  pages={119--130},
  year={1976},
  publisher={Springer},
  doi={10.1007/BF01608499},
  note = {\url{https://doi.org/10.1007/BF01608499}}
}

@inproceedings{CleveWang2017Lindblad,
  author = {Cleve, Richard and Wang, Chunhao},
  title = {Efficient Quantum Algorithms for Simulating {Lindblad} Evolution},
  booktitle = {44th International Colloquium on Automata, Languages, and Programming (ICALP 2017)},
  series = {Leibniz International Proceedings in Informatics (LIPIcs)},
  volume = {80},
  pages = {17:1--17:14},
  publisher = {Schloss Dagstuhl -- Leibniz-Zentrum f{\"u}r Informatik},
  year = {2017},
  doi = {10.4230/LIPIcs.ICALP.2017.17},
  eprint = {1612.09512},
  archivePrefix = {arXiv},
  primaryClass = {quant-ph},
  note = {\url{https://doi.org/10.4230/LIPIcs.ICALP.2017.17}}
}

@book{DaPratoZabczyk2014SPDE,
  author    = {Da Prato, Giuseppe and Zabczyk, Jerzy},
  title     = {Stochastic Equations in Infinite Dimensions},
  edition   = {2},
  series    = {Encyclopedia of Mathematics and its Applications},
  volume    = {152},
  publisher = {Cambridge University Press},
  address   = {Cambridge},
  year      = {2014},
  doi       = {10.1017/CBO9781107295513},
  isbn      = {9781107055841},
  note = {\url{https://doi.org/10.1017/CBO9781107295513}}
}

@book{schlick2010molecular,
  title     = {Molecular Modeling and Simulation: An Interdisciplinary Guide},
  author    = {Schlick, Tamar},
  edition   = {2},
  series    = {Interdisciplinary Applied Mathematics},
  year      = {2010},
  publisher = {Springer},
  doi       = {10.1007/978-1-4419-6351-2},
  note = {\url{https://doi.org/10.1007/978-1-4419-6351-2}}
}

@book{leimkuhler2015molecular,
  title     = {Molecular Dynamics: With Deterministic and Stochastic Numerical Methods},
  author    = {Leimkuhler, Ben and Matthews, Charles},
  series    = {Interdisciplinary Applied Mathematics},
  volume    = {39},
  year      = {2015},
  publisher = {Springer},
  doi       = {10.1007/978-3-319-16375-8},
  note = {\url{https://doi.org/10.1007/978-3-319-16375-8}}
}

@inproceedings{LiWang2023HigherOrder,
  author = {Li, Xiantao and Wang, Chunhao},
  title = {Simulating Markovian Open Quantum Systems Using Higher-Order Series Expansion},
  booktitle = {50th International Colloquium on Automata, Languages, and Programming (ICALP 2023)},
  series = {Leibniz International Proceedings in Informatics (LIPIcs)},
  volume = {261},
  pages = {87:1--87:20},
  publisher = {Schloss Dagstuhl -- Leibniz-Zentrum f{\"u}r Informatik},
  address = {Dagstuhl, Germany},
  year = {2023},
  doi = {10.4230/LIPIcs.ICALP.2023.87},
  note = {\url{https://doi.org/10.4230/LIPIcs.ICALP.2023.87}}
}

@article{DingLiLin2024Hamiltonian,
  author = {Ding, Zhiyan and Li, Xiantao and Lin, Lin},
  title = {Simulating Open Quantum Systems Using Hamiltonian Simulations},
  journal = {PRX Quantum},
  volume = {5},
  pages = {020332},
  year = {2024},
  doi = {10.1103/PRXQuantum.5.020332},
  note = {\url{https://doi.org/10.1103/PRXQuantum.5.020332}}
}

@misc{WuLi2026LinearSDE,
  author = {Wu, Hsuan-Cheng and Li, Xiantao},
  title = {Universal Dilation of Linear It{\^o} SDEs: Quantum Trajectories and {Lindblad} Simulation of Second Moments},
  year = {2026},
  eprint = {2601.05928},
  archivePrefix = {arXiv},
  primaryClass = {quant-ph},
  note = {\url{https://arxiv.org/abs/2601.05928}}
}

@misc{KharaziAlkadriMandadapuWhaley2026ReactionRates,
  author = {Kharazi, Tyler and Alkadri, Ahmad M. and Mandadapu, Kranthi K. and Whaley, K. Birgitta},
  title = {Provable Quantum Speedups for Reaction-Rate Estimation in High-Dimensional Fokker-Planck Dynamics},
  year = {2026},
  eprint = {2601.15523},
  archivePrefix = {arXiv},
  primaryClass = {quant-ph},
  note = {\url{https://arxiv.org/abs/2601.15523}}
}

@misc{WatanabeNishiKosugiHidakaSakuraiMatsushita2026KvNGreenKubo,
  author = {Watanabe, Masari and Nishi, Hirofumi and Kosugi, Taichi and Hidaka, Shigekazu and Sakurai, Ryo and Matsushita, Yu-ichiro},
  title = {{Koopman}--von Neumann Molecular Dynamics for Green--Kubo Transport Coefficients},
  year = {2026},
  eprint = {2605.30142},
  archivePrefix = {arXiv},
  primaryClass = {quant-ph},
  note = {\url{https://arxiv.org/abs/2605.30142}}
}

@article{TangYing2024FunctionalHierarchicalTensor,
  author = {Tang, Xun and Ying, Lexing},
  title = {Solving High-Dimensional Fokker--Planck Equation with Functional Hierarchical Tensor},
  journal = {Journal of Computational Physics},
  volume = {511},
  pages = {113110},
  year = {2024},
  doi = {10.1016/j.jcp.2024.113110},
  note = {\url{https://doi.org/10.1016/j.jcp.2024.113110}}
}

@article{ChenHoskinsKhooLindsey2023CommittorTensor,
  author = {Chen, Yian and Hoskins, Jeremy and Khoo, Yuehaw and Lindsey, Michael},
  title = {Committor Functions via Tensor Networks},
  journal = {Journal of Computational Physics},
  volume = {472},
  pages = {111646},
  year = {2023},
  doi = {10.1016/j.jcp.2022.111646},
  note = {\url{https://doi.org/10.1016/j.jcp.2022.111646}}
}

@misc{ShangGuoRebentrostAspuruGuzikLiZhao2025FastForwardQPE,
  title = {Fast-forwardable {Lindbladian}s imply quantum phase estimation},
  author = {Shang, Zhong-Xia and Guo, Naixu and Rebentrost, Patrick and Aspuru-Guzik, Al{\'a}n and Li, Tongyang and Zhao, Qi},
  year = {2025},
  eprint = {2510.06759},
  archivePrefix = {arXiv},
  primaryClass = {quant-ph},
  note = {\url{https://arxiv.org/abs/2510.06759}}
}

@article{Lax1963Regression,
  author  = {Lax, Melvin},
  title   = {Formal Theory of Quantum Fluctuations from a Driven State},
  journal = {Physical Review},
  volume  = {129},
  number  = {5},
  pages   = {2342--2348},
  year    = {1963},
  doi     = {10.1103/PhysRev.129.2342},
  note = {\url{https://doi.org/10.1103/PhysRev.129.2342}}
}

@misc{gan2025provably,
  title         = {Provably efficient quantum algorithms for solving nonlinear differential equations using multiple bosonic modes coupled with qubits},
  author        = {Gan, Yu and Alipanah, Hirad and Cheng, Jinglei and Wu, Zeguan and Li, Guangyi and Mendoza-Arenas, Juan Jos{\'e} and Givi, Peyman and Malik, Mujeeb R and McDermott, Brian J and Liu, Junyu},
  year          = {2025},
  eprint        = {2511.09939},
  archivePrefix = {arXiv},
  primaryClass  = {quant-ph},
  note = {\url{https://arxiv.org/abs/2511.09939}}
}

@misc{li2025quantum,
  title         = {Quantum Regression Theory and Efficient Computation of Response Functions for Non-Markovian Open Systems},
  author        = {Li, Xiantao and Wang, Chunhao},
  year          = {2025},
  eprint        = {2510.05472},
  archivePrefix = {arXiv},
  primaryClass  = {quant-ph},
  note = {\url{https://arxiv.org/abs/2510.05472}}
}

@article{liu2021efficient,
author = {Jin-Peng Liu  and Herman Øie Kolden  and Hari K. Krovi  and Nuno F. Loureiro  and Konstantina Trivisa  and Andrew M. Childs },
title = {Efficient quantum algorithm for dissipative nonlinear differential equations},
journal = {Proceedings of the National Academy of Sciences},
volume = {118},
number = {35},
pages = {e2026805118},
year = {2021},
doi = {10.1073/pnas.2026805118},
  note = {\url{https://doi.org/10.1073/pnas.2026805118}}
}

@article{joseph2020Koopman,
  title = {Koopman--von Neumann approach to quantum simulation of nonlinear classical dynamics},
  author = {Joseph, Ilon},
  journal = {Phys. Rev. Res.},
  volume = {2},
  number = {4},
  pages = {043102},
  numpages = {17},
  year = {2020},
  month = {Oct},
  publisher = {American Physical Society},
  doi = {10.1103/PhysRevResearch.2.043102},
  note = {\url{https://doi.org/10.1103/PhysRevResearch.2.043102}}
}

@inproceedings{gilyen2019quantum,
  title     = {Quantum Singular Value Transformation and Beyond: Exponential Improvements for Quantum Matrix Arithmetics},
  author    = {Gily{\'e}n, Andr{\'a}s and Su, Yuan and Low, Guang Hao and Wiebe, Nathan},
  booktitle = {Proceedings of the 51st Annual ACM SIGACT Symposium on Theory of Computing},
  pages     = {193--204},
  year      = {2019},
  publisher = {Association for Computing Machinery},
  doi       = {10.1145/3313276.3316366},
  note = {\url{https://doi.org/10.1145/3313276.3316366}}
}

@misc{chakraborty2018power,
  title={The power of block-encoded matrix powers: improved regression techniques via faster Hamiltonian simulation},
  author={Chakraborty, Shantanav and Gily{\'e}n, Andr{\'a}s and Jeffery, Stacey},
  year={2018},
  eprint={1804.01973},
  archivePrefix={arXiv},
  primaryClass={quant-ph},
  note = {\url{https://arxiv.org/abs/1804.01973}}
}

@misc{bravyi2025quantum,
  title={Quantum simulation of a noisy classical nonlinear dynamics},
  author={Bravyi, Sergey and Manson-Sawko, Robert and Zayats, Mykhaylo and Zhuk, Sergiy},
  year={2025},
  eprint={2507.06198},
  archivePrefix={arXiv},
  primaryClass={quant-ph},
  note = {\url{https://arxiv.org/abs/2507.06198}}
}

@misc{bravyi2026quantum,
  title={Quantum algorithms for stochastic nonlinear differential equations},
  author={Bravyi, Sergey and Byrne, Adam and Zayats, Mykhaylo and Zhuk, Sergiy},
  year={2026},
  eprint={2606.08349},
  archivePrefix={arXiv},
  primaryClass={quant-ph},
  note = {\url{https://arxiv.org/abs/2606.08349}}
}

@misc{LiCatliLimPocrnicAnLiuWiebe2026NSDE,
  author = {Li, Xiangyu and Catli, Ahmet Burak and Lim, Ho Kiat and Pocrnic, Matthew and An, Dong and Liu, Jin-Peng and Wiebe, Nathan},
  title = {Efficient Quantum Simulation for Nonlinear Stochastic Differential Equations},
  year = {2026},
  eprint = {2603.12398},
  archivePrefix = {arXiv},
  primaryClass = {quant-ph},
  note = {\url{https://arxiv.org/abs/2603.12398}}
}

@article{wu2025quantum,
    author = {Wu, Hsuan-Cheng and Wang, Jingyao and Li, Xiantao},
    title = {Quantum Algorithms for Nonlinear Dynamics: Revisiting Carleman Linearization with No Dissipative Conditions},
    journal = {SIAM Journal on Scientific Computing},
    volume = {47},
    number = {2},
    pages = {A943-A970},
    year = {2025},
    doi = {10.1137/24M1665799},
  note = {\url{https://doi.org/10.1137/24M1665799}}
}

@misc{jennings2025quantum,
  title         = {Quantum Algorithms for General Nonlinear Dynamics Based on the Carleman Embedding},
  author        = {Jennings, David and Korzekwa, Kamil and Lostaglio, Matteo and Sornborger, Andrew T. and Subasi, Yigit and Wang, Guoming},
  year          = {2025},
  eprint        = {2509.07155},
  archivePrefix = {arXiv},
  primaryClass  = {quant-ph},
  note = {\url{https://arxiv.org/abs/2509.07155}}
}

@article{an2026quantum,
  doi = {10.22331/q-2026-01-19-1969},
  title = {Quantum algorithms for linear and non-linear fractional reaction-diffusion equations},
  author = {An, Dong and Trivisa, Konstantina},
  journal = {{Quantum}},
  issn = {2521-327X},
  publisher = {{Verein zur F{\"{o}}rderung des Open Access Publizierens in den Quantenwissenschaften}},
  volume = {10},
  pages = {1969},
  month = jan,
  year = {2026},
  note = {\url{https://doi.org/10.22331/q-2026-01-19-1969}}
}

@book{breuer2002theory,
  title     = {The Theory of Open Quantum Systems},
  author    = {Breuer, Heinz-Peter and Petruccione, Francesco},
  year      = {2002},
  publisher = {Oxford University Press},
  doi       = {10.1093/acprof:oso/9780199213900.001.0001},
  note = {\url{https://doi.org/10.1093/acprof:oso/9780199213900.001.0001}}
}

@inproceedings{gnanasekaran2023efficient,
   author={Gnanasekaran, Abeynaya and Surana, Amit and Sahai, Tuhin},
  booktitle={2023 IEEE International Conference on Quantum Computing and Engineering (QCE)}, 
  title={Efficient Quantum Algorithms for Nonlinear Stochastic Dynamical Systems}, 
  year={2023},
  volume={02},
  number={},
  pages={66-75},
  doi={10.1109/QCE57702.2023.10186},
  note = {\url{https://doi.org/10.1109/QCE57702.2023.10186}}
}

@article{Koopman1931hamiltonian,
author = {B. O. Koopman },
title = {Hamiltonian Systems and Transformation in Hilbert Space},
journal = {Proceedings of the National Academy of Sciences},
volume = {17},
number = {5},
pages = {315-318},
year = {1931},
doi = {10.1073/pnas.17.5.315},
  note = {\url{https://doi.org/10.1073/pnas.17.5.315}}
}

@article{neumann1932operatorenmethode,
  title   = {Zur Operatorenmethode in der klassischen Mechanik},
  author  = {von Neumann, John},
  journal = {Annals of Mathematics},
  volume  = {33},
  number  = {3},
  pages   = {587--642},
  year    = {1932},
  doi     = {10.2307/1968537},
  note = {\url{https://doi.org/10.2307/1968537}}
}

@inproceedings{Szegedy2004QuantumWalk,
  author    = {Szegedy, Mario},
  title     = {Quantum Speed-Up of Markov Chain Based Algorithms},
  booktitle = {Proceedings of the 45th Annual IEEE Symposium on
               Foundations of Computer Science},
  pages     = {32--41},
  year      = {2004},
  doi       = {10.1109/FOCS.2004.53},
  note = {\url{https://doi.org/10.1109/FOCS.2004.53}}
}

@article{WocjanAbeyesinghe2008QuantumSampling,
  author  = {Wocjan, Pawel and Abeyesinghe, Anura},
  title   = {Speed-up via Quantum Sampling},
  journal = {Physical Review A},
  volume  = {78},
  number  = {4},
  pages   = {042336},
  year    = {2008},
  doi     = {10.1103/PhysRevA.78.042336},
  note = {\url{https://doi.org/10.1103/PhysRevA.78.042336}}
}

@article{ApersSarlette2019QFF,
  author  = {Apers, Simon and Sarlette, Alain},
  title   = {Quantum Fast-Forwarding: Markov Chains and Graph Property Testing},
  journal = {Quantum Information and Computation},
  volume  = {19},
  number  = {3--4},
  pages   = {181--213},
  year    = {2019},
  doi     = {10.26421/QIC19.3-4-1},
  note = {\url{https://doi.org/10.26421/QIC19.3-4-1}}
}

@article{TennieMagri2025FokkerPlanck,
  author  = {Tennie, Felix and Magri, Luca},
  title   = {Integration of the {Fokker--Planck} Equation on Quantum
             Computers: A New Path to Modelling Nonlinear Dynamics},
  journal = {Proceedings of the Royal Society A:
             Mathematical, Physical and Engineering Sciences},
  volume  = {481},
  number  = {2326},
  pages   = {20250016},
  year    = {2025},
  doi     = {10.1098/rspa.2025.0016},
  note = {\url{https://doi.org/10.1098/rspa.2025.0016}}
}

@article{JinLiuYu2026FokkerPlanck,
  author  = {Jin, Shi and Liu, Nana and Yu, Yue},
  title   = {Quantum Simulation of the {Fokker--Planck} Equation via
             {Schr{\"o}dingerization}},
  journal = {Communications in Computational Physics},
  volume  = {39},
  number  = {4},
  pages   = {969--1001},
  year    = {2026},
  doi     = {10.4208/cicp.OA-2024-0097},
  note = {\url{https://doi.org/10.4208/cicp.OA-2024-0097}}
}

@article{Fokker1914,
  author  = {Fokker, A. D.},
  title   = {Die mittlere Energie rotierender elektrischer Dipole
             im Strahlungsfeld},
  journal = {Annalen der Physik},
  volume  = {348},
  number  = {5},
  pages   = {810--820},
  year    = {1914},
  doi     = {10.1002/andp.19143480507},
  note = {\url{https://doi.org/10.1002/andp.19143480507}}
}

@article{TronciJoseph2021Koopman,
  author  = {Tronci, Cesare and Joseph, Ilon},
  title   = {{Koopman} Wavefunctions and Clebsch Variables in
             Vlasov--Maxwell Kinetic Theory},
  journal = {Journal of Plasma Physics},
  volume  = {87},
  number  = {4},
  pages   = {835870402},
  year    = {2021},
  doi     = {10.1017/S0022377821000805},
  note = {\url{https://doi.org/10.1017/S0022377821000805}}
}

@article{Montanaro2015QuantumMonteCarlo,
  author  = {Montanaro, Ashley},
  title   = {Quantum Speedup of Monte Carlo Methods},
  journal = {Proceedings of the Royal Society A:
             Mathematical, Physical and Engineering Sciences},
  volume  = {471},
  number  = {2181},
  pages   = {20150301},
  year    = {2015},
  doi     = {10.1098/rspa.2015.0301},
  note = {\url{https://doi.org/10.1098/rspa.2015.0301}}
}

@article{Planck1917,
  author  = {Planck, Max},
  title   = {{\"U}ber einen Satz der statistischen Dynamik und seine
             Erweiterung in der Quantentheorie},
  journal = {Sitzungsberichte der K{\"o}niglich Preussischen Akademie
             der Wissenschaften zu Berlin,
             Physikalisch-Mathematische Klasse},
  pages   = {324--341},
  year    = {1917},
  url = {https://edoc.bbaw.de/opus4-bbaw/frontdoor/deliver/index/docId/5368/file/BBAW_SB_1917_TB1_S324_341.pdf}
}

@article{Kolmogorov1931,
  author  = {Kolmogoroff, A. N.},
  title   = {{\"U}ber die analytischen Methoden in der
             Wahrscheinlichkeitsrechnung},
  journal = {Mathematische Annalen},
  volume  = {104},
  number  = {1},
  pages   = {415--458},
  year    = {1931},
  doi     = {10.1007/BF01457949},
  note = {\url{https://doi.org/10.1007/BF01457949}}
}

@article{Watrous2009MixedUnitary,
  author  = {Watrous, John},
  title   = {Mixing Doubly Stochastic Quantum Channels with the
             Completely Depolarizing Channel},
  journal = {Quantum Information and Computation},
  volume  = {9},
  number  = {5--6},
  pages   = {406--413},
  year    = {2009},
  doi     = {10.26421/QIC9.5-6-4},
  note = {\url{https://doi.org/10.26421/QIC9.5-6-4}}
}

@article{LeeWatrous2020MixedUnitary,
  author  = {Lee, Colin Do-Yan and Watrous, John},
  title   = {Detecting Mixed-Unitary Quantum Channels Is {NP}-Hard},
  journal = {Quantum},
  volume  = {4},
  pages   = {253},
  year    = {2020},
  doi     = {10.22331/q-2020-04-16-253},
  note = {\url{https://doi.org/10.22331/q-2020-04-16-253}}
}

@article{IrrgeherKritzerPillichshammerWozniakowski2016Hermite,
  author  = {Irrgeher, Christian and Kritzer, Peter and
             Pillichshammer, Friedrich and Wo{\'z}niakowski, Henryk},
  title   = {Approximation in Hermite Spaces of Smooth Functions},
  journal = {Journal of Approximation Theory},
  volume  = {207},
  pages   = {98--126},
  year    = {2016},
  doi     = {10.1016/j.jat.2016.02.008},
  note = {\url{https://doi.org/10.1016/j.jat.2016.02.008}}
}

@article{IrrgeherKritzerPillichshammerWozniakowski2016Exponential,
  author  = {Irrgeher, Christian and Kritzer, Peter and
             Pillichshammer, Friedrich and Wo{\'z}niakowski, Henryk},
  title   = {Tractability of Multivariate Approximation Defined over
             Hilbert Spaces with Exponential Weights},
  journal = {Journal of Approximation Theory},
  volume  = {207},
  pages   = {301--338},
  year    = {2016},
  doi     = {10.1016/j.jat.2016.02.020},
  eprint  = {1502.03286},
  archivePrefix = {arXiv},
  note = {\url{https://doi.org/10.1016/j.jat.2016.02.020}}
}

@article{LeobacherPillichshammerEbert2023Hermite,
  author  = {Leobacher, Gunther and Pillichshammer, Friedrich and
             Ebert, Adrian},
  title   = {Tractability of {$L_2$}-Approximation and Integration in
             Weighted Hermite Spaces of Finite Smoothness},
  journal = {Journal of Complexity},
  volume  = {78},
  pages   = {101768},
  year    = {2023},
  doi     = {10.1016/j.jco.2023.101768},
  note = {\url{https://doi.org/10.1016/j.jco.2023.101768}}
}

@article{LuoYau2013HermiteHyperbolicCross,
  author  = {Luo, Xue and Yau, Stephen S.-T.},
  title   = {Hermite Spectral Method with Hyperbolic Cross
             Approximations to High-Dimensional Parabolic {PDE}s},
  journal = {SIAM Journal on Numerical Analysis},
  volume  = {51},
  number  = {6},
  pages   = {3186--3212},
  year    = {2013},
  doi     = {10.1137/120896931},
  note = {\url{https://doi.org/10.1137/120896931}}
}

@article{Giannakis2019QuantumDA,
  author  = {Giannakis, Dimitrios},
  title   = {Quantum Mechanics and Data Assimilation},
  journal = {Physical Review E},
  volume  = {100},
  number  = {3},
  pages   = {032207},
  year    = {2019},
  doi     = {10.1103/PhysRevE.100.032207},
  note = {\url{https://doi.org/10.1103/PhysRevE.100.032207}}
}

@article{GiannakisEtAl2022Embedding,
  author  = {Giannakis, Dimitrios and Ourmazd, Abbas and Pfeffer, Philipp and
             Schumacher, J{\"o}rg and S{\l}awi{\'n}ska, Joanna},
  title   = {Embedding Classical Dynamics in a Quantum Computer},
  journal = {Physical Review A},
  volume  = {105},
  number  = {5},
  pages   = {052404},
  year    = {2022},
  doi     = {10.1103/PhysRevA.105.052404},
  note = {\url{https://doi.org/10.1103/PhysRevA.105.052404}}
}

@article{LengDingChenLin2026OperatorSampling,
  author = {Leng, Jiaqi and Ding, Zhiyan and Chen, Zherui and Lin, Lin},
  title = {Operator-level quantum acceleration of non-logconcave sampling},
  journal = {Proceedings of the National Academy of Sciences},
  volume = {123},
  number = {8},
  pages = {e2512789123},
  year = {2026},
  doi = {10.1073/pnas.2512789123},
  eprint = {2505.05301},
  archivePrefix = {arXiv},
  primaryClass = {quant-ph},
  note = {\url{https://doi.org/10.1073/pnas.2512789123}}
}

\end{document}